\documentclass[10pt,conference]{IEEEtran}
\usepackage{fixltx2e}
\usepackage{cite}
\usepackage{url}
\usepackage{mathrsfs}
\usepackage{color}
\usepackage{float}
\usepackage[caption=false]{subfig}
\usepackage{balance}
\ifCLASSINFOpdf
   \usepackage[pdftex]{graphicx}
   \graphicspath{{Figs/}}
   \DeclareGraphicsExtensions{.pdf,.jpeg,.png}
\else
\fi

\usepackage[cmex10]{amsmath}
\usepackage{amsmath}
\usepackage{amssymb}
\usepackage{amsthm}
\usepackage{amsfonts}
\usepackage{bm}
\usepackage{xfrac}
\usepackage{empheq}
\usepackage[normalem]{ulem} % either use this (simple) or
\usepackage{soul} % use this (many fancier options)
\usepackage{mathtools}
\DeclarePairedDelimiter\ceil{\lceil}{\rceil}
\DeclarePairedDelimiter\floor{\lfloor}{\rfloor}

\newtheorem{theorem}{Theorem}

\newtheorem{example}{Example}

\newtheorem{lemma}{Lemma}

\newtheorem{remark}{Remark}
\theoremstyle{definition}

\usepackage[a4paper,bindingoffset=0.2in,%
left=1in,right=1in,top=1in,bottom=1in,%
footskip=.25in]{geometry}
\graphicspath{{figs/}}
\newenvironment{sproof}{\noindent{ \emph{ Sketch of proof:}}}{\qed\bigskip}
\begin{document}
	\newgeometry{left=0.7in,right=0.7in,top=.709in,bottom=1.02in}
	\title{Cache-Aided Variable-Length Coding with Perfect Privacy}
	%\title{Bounds for Multi-User Privacy-Utility Trade-off with Non-zero Leakage}
\vspace{-5mm}
\author{
		\IEEEauthorblockN{Amirreza Zamani, Mikael Skoglund \vspace*{0.5em}
			\IEEEauthorblockA{\\
                              Division of Information Science and Engineering, KTH Royal Institute of Technology \\
                              %$^\ddagger$Dept. of Electrical and Electronic Engineering, Imperial College London\\
				Email: \protect amizam@kth.se, skoglund@kth.se }}%\vspace*{-2.1em}
		}
	\maketitle
%
%\iffalse
\begin{abstract} 
	A cache-aided compression problem with perfect privacy is studied, where a server has access to a database of $N$ files, $(Y_1,...,Y_N)$, each of size $F$ bits. The server is connected to $K$ users through a shared link, 
	where each user has access to a local cache of size $MF$ bits. In the placement phase, the server fills the users$'$ caches without prior knowledge of their future demands, while the delivery phase takes place after the users send their demands to the server. 
	We assume that each file $Y_i$ is arbitrarily correlated with a private attribute $X$, and an adversary is assumed to have access to the shared link. The users and the server have access to a shared secret key $W$.  
	The goal is to design the cache contents and the delivered message $\cal C$ 
	such that the average length of $\mathcal{C}$ is minimized, while satisfying: 
i. The response $\cal C$ does not disclose any information about $X$, i.e., $X$ and $\cal C$ are statistically independent yielding $I(X;\mathcal{C})=0$, which corresponds to the perfect privacy constraint; ii. User $i$ is able to decode its demand, $Y_{d_i}$, by using its local cache $Z_i$, delivered message $\cal C$, and the shared secret key $W$.
	 Due to the correlation of database with the private attribute, existing codes for cache-aided delivery do not fulfill the perfect privacy constraint. Indeed, in this work, we propose a lossless variable-length coding scheme that combines privacy-aware compression with coded caching techniques. In particular, we use two-part code construction and Functional Representation Lemma. %Finally, we extend the results to the case, where $X$ and $\mathcal{C}$ can be correlated, i.e., non-zero leakage is allowed. 
	 Furthermore, we propose an alternative coding scheme based on the minimum entropy coupling concept and a greedy entropy-based algorithm. We show that the proposed scheme improves the previous results obtained by Functional Representation Lemma. Considering two special cases we improve both coding schemes using the common information concept. Finally, we compare the proposed schemes in numerical examples and provide an application considering an encoder with limited buffer size.  
\end{abstract}
\section{Introduction}
We consider the system model depicted in Fig. \ref{wii}, wherein a central server has access to a database consisting of $N$ files, denoted by $Y_1, \ldots, Y_N$, each of size $F$ bits. These files are jointly distributed according to a probability distribution $P_{X Y_1 \cdots Y_N}$, where $X$ represents a private latent variable. It is assumed that the realization of the private variable $X$ is known to the server. The server communicates with $K$ users through a shared broadcast link. Each user $k \in [K] \triangleq \{1, \ldots, K\}$ is equipped with a local cache of capacity $MF$ bits. Additionally, the server and the users share a common secret key $W$ of size $T$. 
The system operates in two distinct phases: the placement phase and the delivery phase, in accordance with the coded caching framework proposed in \cite{maddah1}. During the placement phase, the server fills the users’ local cache memories using the database. Let $Z_k$ denote the cache content of user $k$, at the end of this phase. In the delivery phase, each user sends a file request, where $d_k \in [N]$ denotes the index of the file requested by user $k$. The server then communicates a response denoted by $\mathcal{C}$, over the shared link, designed to satisfy all user demands simultaneously.
 We assume that an adversary has access to the shared link as well, observes $\cal C$ and aims to infer information about the latent variable $X$. The adversary is assumed to have access only to $\mathcal{C}$ and neither the secret key $W$ nor the cache contents $\{Z_k\}_{k=1}^K$.
 Due to the statistical dependence between the files and the latent variable $X$, the coded caching and delivery techniques as proposed in \cite{maddah1} are insufficient to ensure privacy.
 The cache-aided delivery with perfect privacy problem aims to design the response $\mathcal{C}$ with minimum possible average length while satisfying two constraints: (i) a perfect privacy constraint, requiring that $\mathcal{C}$ is statistically independent of $X$, and (ii) a zero-error decoding constraint, which mandates that each user $k$ can reconstruct its requested file $Y_{d_k}$ without error using its cache contents $Z_k$, the shared secret key $W$, and the server’s response $\mathcal{C}$, i.e., $\mathbb{P}\{\hat{Y}_{d_k} \neq Y_{d_k}\} = 0$, for all $k \in [K]$ where $\hat{Y}_{d_k}$ denote the decoded message of user $k$ using $W$, $\cal C$, and $Z_k$. The design of $\mathcal{C}$ is based on the worst-case demand combinations $\mathbf{d} = (d_1, \ldots, d_K)$, and the expectation is taken over the joint distribution of the database files and the latent variable. 
 The goal of the cache-aided private delivery problem is to find a response $\mathcal{C}$ with minimum possible average length that satisfies a certain privacy constraint and the zero-error decodability constraint of users. Here, we consider the worst case demand combinations $d=(d_1,..,d_K)$ to construct $\cal C$, and the expectation is taken over the randomness in the database. 
 \begin{figure}[]
 	\centering
 	\includegraphics[scale = .12]{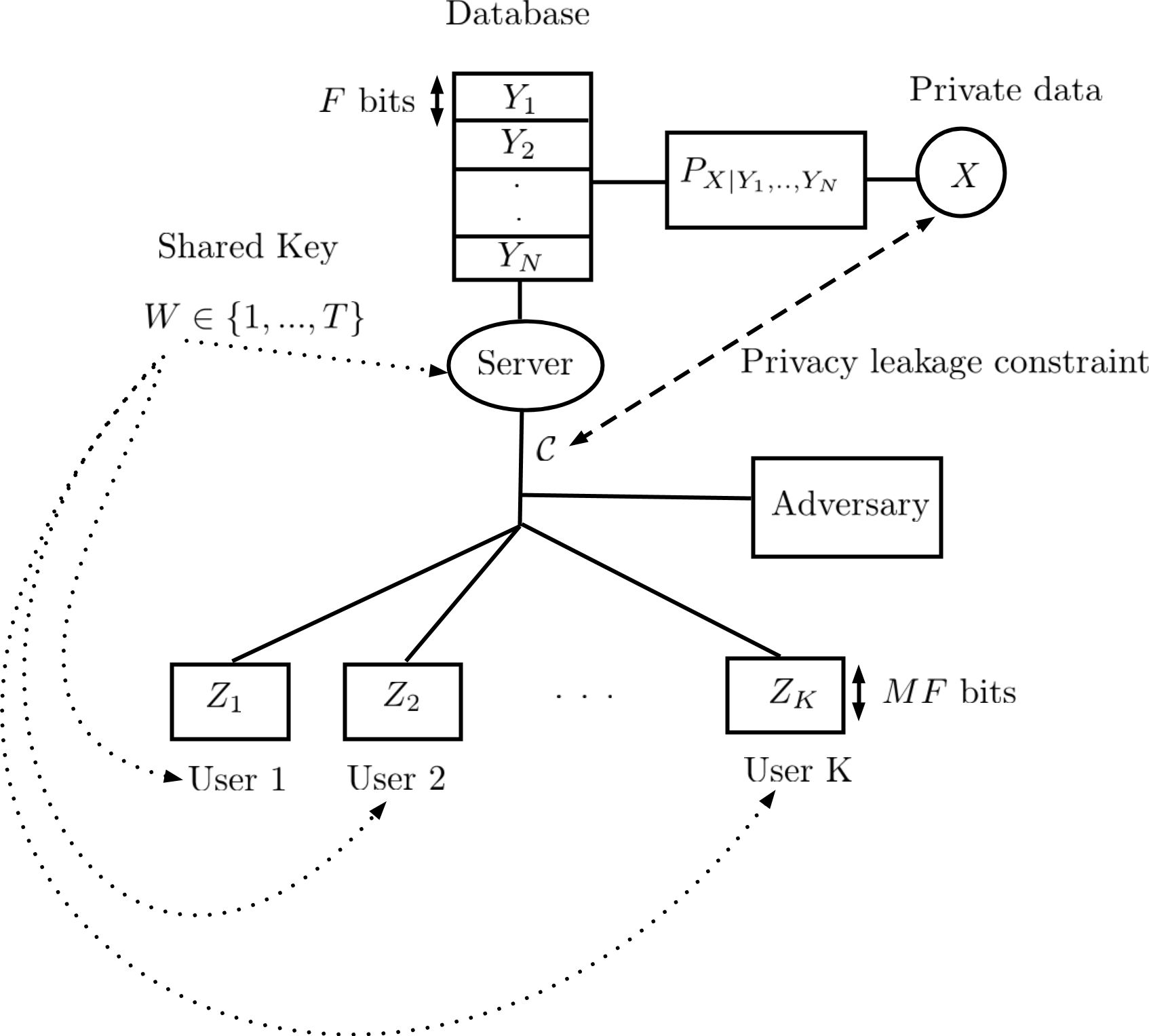}
 	\caption{In this work a server wants to send a response via a shared link to fulfill users$'$ requests, but due to the correlation of database with the private data existing schemes cannot be used. In the delivery phase, we hide the information about $X$ using one-time-pad coding and communicate the rest of response using different techniques.}
 	%In this work an encoder wants to compress $Y$ which is correlated with $X$ under certain privacy leakage constraints and send it over a channel where an eavesdropper has access to the output of the encoder. The encoder and decoder have an advantage of using shared secret key.}
 	\label{wii}
 \end{figure} 
 %In this work, we first consider a perfect privacy constraint, i.e., we require $\cal C$ to be independent of $X$. Let $\hat{Y}_{d_k}$ denote the decoded message of user $k$ using $W$, $\cal C$, and $Z_k$. User $k$ should be able to recover $Y_{d_k}$ reliably, i.e., $\mathbb{P}{\{\hat{Y}_{d_k}\neq Y_{d_k}\}}=0$, $\forall k\in[K]$. 
We remark that, we have a cache-aided variable-length compression problem with a privacy constraint. Accordingly, our approach utilizes tools and techniques from multiple domains, including privacy-preserving mechanism design, information-theoretic source compression, and cache design and coded delivery problems. Specifically, our coding scheme draws upon variable-length and lossless compression techniques developed in \cite{kostala}, and incorporates cache placement and delivery strategies proposed by the coded caching framework in \cite{maddah1}.
Furthermore, we propose a new coding scheme to build $\cal C$, where we use the minimum entropy coupling concept and a greedy entropy-based algorithm that are studied in \cite{kocaoglu2017entropic, compton2023minimum, shkel2023information}. 
We compare the new proposed scheme with the existing one and show that the proposed achievable scheme can significantly improve the previous result. 

The problems of privacy mechanism design, data compression, and cache-aided content delivery have received increased attention in recent years
%The privacy mechanism design problem is receiving increased attention in information theory recently. 
%Related works can be found in 
\cite{maddah1, lim,kocaoglu2017entropic, compton2023minimum, shkel2023information, wang33,  lu, denizjadid, shannon, dworkal, dwork1, denizjadid2, gunduz2010source, schaefer, yamamoto1988rate, yamamoto, cufff, Calmon2, makhdoumi, borz, khodam, Khodam22,kostala, kostala2, calmon4, asoo, issa2, shah, asoodeh1,deniz3,privatecache,Ali,demand,10464318,zero,nonzero,zamani2025variable}. 
In particular, the seminal work in \cite{maddah1} considers a cache-aided broadcast network, where a single server communicates with multiple users via a shared, error-free bottleneck link. The fundamental trade-off between cache memory and delivery rate known as the rate-memory trade-off, was characterized up to a constant multiplicative gap. Subsequent works have improved these bounds for various settings. For instance, tighter bounds have been established in \cite{lim, wang33}, and the exact rate-memory trade-off for the case of uncoded cache placement was derived in \cite{lu}. Moreover, the work in \cite{denizjadid} extends the coded caching model to scenarios in which users have heterogeneous quality-of-service requirements, by considering lossy content delivery under different distortion constraints. The problem of coded caching with private caches is studied in \cite{privatecache}. Coded caching with simultaneously private caches and demands is studied in \cite{Ali}. In \cite{demand}, coded caching with demand and file privacy is studied. In contrast to previous works, which focus on the privacy of files, user demands, and cache contents, the emphasis here is on analyzing the leakage of the private attribute $X$, which is arbitrarily correlated with the dataset, to the adversary.
%In \cite{gholami}, the problem of coded caching with private demands and caches has been studied, where the private coded caching schemes are constructed using private information retrieval (PIR).

In the context of source compression under privacy constraints, the notion of perfect secrecy was introduced by Shannon in his seminal work \cite{shannon}, wherein public and private data are required to be statistically independent. In the classical Shannon cipher system, the goal is to transmit one of $M$ possible messages over a communication channel that is observed by an eavesdropper. Shannon proved that perfect secrecy is achievable if and only if the shared secret key has entropy at least equal to the entropy of the message, i.e., the key length must be no less than $\log M$ bits \cite{shannon}.

The requirement of perfect secrecy, characterized by the independence of disclosed (public) data and private information has also been adopted in the differential privacy literature, notably in \cite{dworkal, dwork1}, where sanitized versions of sensitive databases are released for public use. Equivocation, defined as the conditional entropy of the private data given the public output, has been widely used as a metric for quantifying information leakage in information-theoretic security, as in \cite{denizjadid2, gunduz2010source, schaefer}. A rate-distortion approach to information-theoretic secrecy is studied in \cite{yamamoto1988rate} and a related source coding problem with secrecy is studied in \cite{yamamoto}.  
%The concept of maximal leakage has been introduced in \cite{issa} and some bounds on the privacy utility trade-off have been derived. 
Considering privacy design problems through the lens of information theory, 
fundamental limits of the privacy utility trade-off measuring the leakage using estimation-theoretic guarantees are studied in \cite{Calmon2}.
In \cite{makhdoumi}, the concept of a privacy funnel is introduced, where the privacy-utility trade-off has been studied using the log-loss as a privacy measure as well as a distortion measure for utility. %and the log-loss as 
In \cite{yamamoto}, the privacy-utility trade-off considering equivocation and expected distortion as measures of privacy and utility are studied.
In \cite{borz}, the problem of privacy-utility trade-off considering mutual information both as measures of utility and privacy is studied. It is shown that under the perfect privacy assumption, the optimal privacy mechanism problem can be obtained as the solution of a linear program. %This has been extended in \cite{gun} considering the privacy utility trade-off with a rate constraint on the disclosed data.
%Moreover, in \cite{borz}, it has been shown that information can be only revealed if the kernel (leakage matrix) between useful data and private data is not invertible.
In \cite{khodam}, the work \cite{borz} is extended by relaxing the perfect privacy assumption allowing some small bounded leakage. Specifically, privacy mechanisms with a per-letter (point-wise) privacy criterion considering an invertible kernel are designed allowing a bounded leakage. In \cite{Khodam22}, this result is generalized to a non-invertible leakage matrix.
In \cite{kostala}, \emph{secrecy by design} problem is studied under the perfect secrecy assumption. Bounds on secure decomposition have been derived using the Functional Representation Lemma. These results are obtained under the perfect secrecy assumption.
%In \cite{kostala}, the problem of \emph{secrecy by design} is studied and bounds on privacy-utility trade-off for two scenarios where the private data is hidden or observable are derived by using the Functional Representation Lemma. These results are derived under the perfect secrecy assumption, i.e., no leakages are allowed. %The bounds are tight when the private data is a deterministic function of the useful data.
In \cite{shah}, the privacy problems considered in \cite{kostala} are generalized by relaxing the perfect secrecy constraint and allowing some leakages. %More specifically, we considered bounded mutual information, i.e., $I(U;X)\leq \epsilon$ for privacy leakage constraint.
%Furthermore, in the special case of perfect privacy we derived a new upper bound for the perfect privacy function and it has been shown that this new bound generalizes the bound in \cite{kostala}. Moreover, it has been shown that the bound is tight when $|\mathcal{X}|=2$.\\
Furthermore, in \cite{shah}, the privacy-utility trade-off with two different per-letter privacy constraints is studied. %Upper and lower bounds are derived and it has been shown that the bounds in the first scenario, where the private data is hidden to the agent, are asymptotically optimal when the private data is a deterministic function of useful data. \\
In \cite{kostala}, both fixed-length and variable-length source compression problems are studied under privacy constraints. Upper and lower bounds on the expected length of the encoded message are obtained, assuming that the encoded output is statistically independent of the private data, thereby ensuring perfect secrecy. Extending this line of work, \cite{kostala2} studies the trade-offs among secrecy, compression rate, and the size of the shared key in the context of lossless compression. This study considers different privacy metrics, including perfect secrecy, secrecy-by-design, maximal leakage, mutual information leakage, and local differential privacy, and characterizes their implications on compression performance under strong privacy guarantees. In \cite{zero}, \cite{nonzero} and \cite{zamani2025variable}, we have extended the compression problems with secrecy constraints considered in \cite{kostala} and \cite{kostala2}. Specifically, in \cite{zero}, we have improved the bounds in \cite{kostala} in two cases using the concept of \emph{common information} and proposing new achievability schemes. Furthermore, in \cite{nonzero}, we have generalized the results in \cite{kostala} and \cite{kostala2} by by relaxing the perfect secrecy assumption and allowing a bounded leakage.
%Our problem here is closely related to \cite{maddah1} and \cite{kostala}, where in \cite{kostala}, for the problem of lossless data compression, strong information theoretic guarantees are provided and fundamental limits are characterized when the private data is a deterministic function of the useful data. 

The conference versions regarding this paper can be found in \cite{amircache} and \cite{amircache2}.\\
The main contribution of this work is to generalize the 
problem considered in \cite{maddah1} by considering correlation between the database and the private latent data.
Our contribution can be summarized as follows:\\
\textbf{(i)} In Section \ref{sec:system}, we define the variable-length compression problem under the perfect privacy constraint and provide its relevance to the prior works. \\ 
\textbf{(ii)} In Section \ref{background}, we provide a brief background and essential lemmas to build different coding schemes.\\
\textbf{(iii)} In Section \ref{sec:resul}, we 
propose three coding schemes as follows:
\begin{itemize}
	\item[$\bullet$] We use variable-length, lossless compression techniques inspired by \cite{kostala} to propose a novel approach to cache-aided content delivery in the presence of an adversary. The proposed method ensures privacy while enabling efficient delivery. We develop a coding scheme using a two-part code structure that uses the Functional Representation Lemma (FRL) along with one-time pad code. This design ensures that information about the private attribute $X$ is concealed from the adversary, while enabling users to recover their requested files without error.
	\item[$\bullet$] To enhance the performance of the initial construction, we introduce an alternative two-part code, based on a greedy entropy-based algorithm proposed in \cite{kocaoglu2017entropic}. This approach is again combined with one-time pad coding to provide improved bounds on the average length of the code.
	\item[$\bullet$]For two specific scenarios, we derive improved bounds by using the concept of common information. Notably, we show that when the size of the private data is large, the obtained bounds are improved and the required size of the shared secret key can be significantly reduced, while still achieving the desired privacy and decodability constraints. 
\end{itemize}
Furthermore, we compare the obtained bounds in numerical examples.\\
\textbf{(iv)} In Section \ref{st}, we provide an application considering an encoder equipped by a local buffer with limited size and the server sends its response to the encoder. The encoder sequentially encodes the response to build a message and transmits it via the shared link.\\
%use
%variable-length lossless compression techniques as in \cite{kostala} to find an alternative solution to cache-aided data delivery problems in the presence of an adversary. The key idea is to use a two-part code construction, which is based on the Functional Representation Lemma (FRL) and one-time pad coding to hide the information about $X$ and reconstruct the demanded files at user side. 
%Next, we improve the coding scheme which is based on the two-part code construction and the FRL. To do so we use another two-part code construction, which is based on the greedy entropy-based algorithm proposed in \cite{kocaoglu2017entropic} and one-time pad coding to hide the information about $X$ and reconstruct the demanded files at user side. Finally, considering two special cases, we improve the obtained bounds by using the common information concept. We show that when the size of the private data is large, the obtained bounds can be significantly improved using less shared key size. \\
%\textbf{(iv)} In Section \ref{example}, we provide a numerical example based on the MNIST data set to compare the obtained bounds and study the tightness.\\
%\textbf{(v)} In Section \ref{app}, we study an application considering a semantic communication with privacy constraints and two separate blind encoders.\\
The paper is concluded in Section \ref{concul}.

%Finally, we generalize the results for non-zero leakage, i.e., $X$ and $\mathcal{C}$ are allowed to be correlated. The key idea for generalizing to the non-zero leakage case is to extend FRL considering correlation between $X$ and the output that has been done in \cite[Lemma 3]{shah}.      

\section{system model and Problem Formulation} \label{sec:system}
Let $Y_i$ denote the $i$-th file in the database, which is of size $F$ bits, i.e., $\mathcal{Y}_i\in\{1,\ldots,2^F\}$ and $|\mathcal{Y}_i|=2^F$. In this work, we assume that $N\geq K$; however, the results can be extended to other cases as well. Let the discrete random variable (RV) $X$ defined on the finite alphabet $\cal{X}$ describe the private latent variable and be arbitrarily correlated with the files in the database $Y=(Y_1,\ldots,Y_N)$ where $|\mathcal{Y}|= |\mathcal{Y}_1|\times\ldots\times|\mathcal{Y}_N|=(2^F)^N$ and $\mathcal{Y}= \mathcal{Y}_1\times\ldots\times\mathcal{Y}_N$. %$|\mathcal{Y}|= (2^F)^N$. 
We show the joint distribution of the database and the private data by $P_{XY_1\cdot Y_N}$ and marginal distributions of $X$ and $Y_i$ by vectors $P_X$ and $P_{Y_i}$ defined on $\mathbb{R}^{|\mathcal{X}|}$ and $\mathbb{R}^{2^F}$ given by the row and column sums of $P_{XY_1\cdot Y_N}$. The relation between $X$ and $Y$ is given by the matrix $P_{X|Y_1\cdot Y_N}$ defined on $\mathbb{R}^{|\mathcal{X}|\times(2^F)^N}$.
We emphasize that each user has access to a local cache of size $MF$ bits.
The shared secret key is denoted by the discrete RV $W$ defined on $\{1,\ldots,T\}$, and is assumed to be known by the server and the users, but not the adversary. Furthermore, we assume that $W$ is uniformly distributed and is independent of $X$ and $Y$. %Let $Y=(Y_1,...,Y_N)$ where $|\mathcal{Y}|\leq (2^F)^N$ and let $[K]=\{1,..,K\}$. 
Similarly to \cite{maddah1}, we have $K$ caching functions to be used during the placement phase:
\begin{align}
\theta_k: [|\mathcal{Y}|] \rightarrow [2^{\floor{FM}}],\ \forall k\in[K], 
\end{align} 
such that
\begin{align}
Z_k=\theta_k(Y_1,\ldots,Y_N),\ \forall k\in[K].
\end{align} 
Let the vector $(Y_{d_1},\ldots,Y_{d_K})$ denote the demands sent by the users at the beginning of the delivery phase, where $(d_1,\ldots,d_K)\in[N]^K$.  
A variable-length prefix-free code with a shared secret key of size $T$ consists of mappings:
\begin{align*}
&(\text{encoder}) \ \mathcal{C}: \ [|\mathcal{Y}|]\times [T]\times[N]^K \rightarrow \{0,1\}^*\\
&(\text{decoder}) \mathcal{D}_k\!: \! \{0,1\}^*\!\!\times\! [T]\!\times\! [2^{\floor{MF}}]\!\times\! [N]^K\!\!\!\rightarrow\! 2^F\!\!\!,\ \! \forall k\!\in\![K].
\end{align*}
The output of the encoder $\mathcal{C}(Y,W,d_1,\ldots,d_K)$ is the codeword the server sends over the shared link in order to satisfy the demands of the users $(Y_{d_1},\ldots,Y_{d_K})$. At the user side, user $k$ employs the decoding function $\mathcal{D}_k$ to recover its demand $Y_{d_k}$, i.e., $\hat{Y}_{d_k}=\mathcal{D}_k(Z_k,W,\mathcal{C}(Y,W,d_1,\ldots,d_K),d_1,\ldots,d_K)$.
Since the code is prefix-free, no codeword in the image of $\cal C$ is a prefix of another codeword. The variable-length code $(\mathcal{C},\mathcal{D}_1,..,\mathcal{D}_K)$ is lossless if for all $k\in[K]$ we have
\begin{align}\label{choon}
\mathbb{P}(\mathcal{D}_k(\mathcal{C}(Y,W,d_1,\ldots,d_K),W,Z_k,d_1,\ldots,d_K)\!=\!Y_{d_k})\!=\!1.
\end{align} 
In the following, we define perfectly private codes.
The code $(\mathcal{C},\mathcal{D}_1,\ldots,\mathcal{D}_K)$ is \textit{perfectly private} if
\begin{align}
I(\mathcal{C}(Y,W,d_1,\ldots,d_K);X)=0,\label{lashwi}
\end{align}
%and the code is 
%\textit{$\epsilon$-private} if
%\begin{align}
%I(\mathcal{C}(Y,W,d_1,\ldots,d_K);X)=\epsilon.\label{lash1}
%\end{align}
Let $\xi$ be the support of $\mathcal{C}(Y,W,d_1,\ldots,d_K)$, where $\xi\subseteq \{0,1\}^*$. For any $c\in\xi$, let $\mathbb{L}(c)$ be the length of the codeword. The lossless code $(\mathcal{C},\mathcal{D}_1,\ldots,\mathcal{D}_K)$ is \textit{$(\alpha,T,d_1,\ldots,d_K)$-variable-length} if 
\begin{align}\label{jojowi}
\mathbb{E}(\mathbb{L}(\mathcal{C}(Y,w,d_1,\ldots,d_K)))\!\leq\! \alpha,\ \forall w\!\in\!\![T]\ \text{and}\ \forall d_1,\ldots,d_K,
\end{align} 
and $(\mathcal{C},\mathcal{D}_1,\ldots,\mathcal{D}_K)$ satisfies \eqref{choon}.
Finally, let us define the set $\mathcal{H}(\alpha,T,d_1,\ldots,d_K)$ as follows:\\
$\mathcal{H}(\alpha,T,d_1,\ldots,d_K)\triangleq\{(\mathcal{C},\mathcal{D}_1,\ldots,\mathcal{D}_K): (\mathcal{C},\mathcal{D}_1,\ldots,\mathcal{D}_K)\ \text{is}\ \text{perfectly-private and}\\ (\alpha,T,d_1,\ldots,d_K)\text{-variable-length}  \}$. %and %$\mathcal{H}^{\epsilon}(\alpha,T,d_1,\ldots,d_K)\triangleq\{(\mathcal{C},\mathcal{D}_1,\ldots,\mathcal{D}_K): (\mathcal{C},\mathcal{D}_1,\ldots,\mathcal{D}_K)\ \text{is}\ \epsilon\text{-private and}\\ (\alpha,T,d_1,\ldots,d_K)\text{-variable-length}  \}.$
The cache-aided compression design problem with perfect privacy can be stated as follows
\begin{align}
\mathbb{L}(P_{XY_1\cdot Y_N},T)&=\!\!\!\!\!\inf_{\begin{array}{c} 
	\substack{d_1,\ldots,d_K,(\mathcal{C},\mathcal{D}_1,\ldots,\mathcal{D}_K)\in\mathcal{H}(\alpha,T,d_1,\ldots,d_K)}
	\end{array}}\alpha.\label{main1wi}
%\mathbb{L}^{\epsilon}(P_{XY_1\cdot Y_N},T)&=\!\!\!\!\!\!\!\!\!\!\inf_{\begin{array}{c} 
%	\substack{d_1,\ldots,d_K,(\mathcal{C},\mathcal{D}_1,\ldots,\mathcal{D}_K)\in\mathcal{H}^{\epsilon}(\alpha,T,d_1,\ldots,d_K)}
%	\end{array}}\alpha.\label{main2wi}
\end{align} 
\begin{remark}
	\normalfont 
	By letting $M=0$, i.e., 
	local caches do not exist, $N=1$, and $K=1$, \eqref{main1wi} leads to the privacy-compression rate trade-off studied in \cite{kostala}, \cite{zero}, \cite{nonzero} and \cite{zamani2025variable}, where upper and lower bounds have been derived.
	%statistically independence between the encoded message and the private data, both \eqref{main1} and \eqref{main3} lead to the privacy-compression rate trade-off studied in \cite{kostala}, where upper and lower bounds have been derived. In this paper, we generalize the trade-off by considering non-zero $\epsilon$.
\end{remark}
\begin{remark}
	\normalfont 
	By letting $M=0$, and considering $X=(X_1,\ldots,X_N)$, a similar privacy-utility trade-off has been studied in \cite{zamani2022multi}, where each user wants a subvector of the database $Y=(Y_1,\ldots,Y_N)$ that is correlated with $X=(X_1,\ldots,X_N)$ and the server maximizes a linear combination of utilities. Using the Functional Representation Lemma and Strong Functional Representation Lemma, upper and lower bounds have been derived, which are shown to be tight within a constant. 
	%statistically independence between the encoded message and the private data, both \eqref{main1} and \eqref{main3} lead to the privacy-compression rate trade-off studied in \cite{kostala}, where upper and lower bounds have been derived. In this paper, we generalize the trade-off by considering non-zero $\epsilon$.
\end{remark}
\begin{remark}
		\normalfont
	In this paper, to design a coding scheme, we consider the worst case demand combinations. Latter follows since \eqref{jojowi} must hold for all possible combinations of the demands.
\end{remark}
\section{Preliminaries and Background}\label{background}
Here, we present the essential results to obtain the design of the coding schemes.
 \begin{lemma}\label{FRL}(FRL \cite[Lemma~1]{kostala}):
	For any pair of RVs $(X,Y)$ distributed according to $P_{XY}$ supported on alphabets $\mathcal{X}$ and $\mathcal{Y}$, respectively, where $|\mathcal{X}|$ is finite and $|\mathcal{Y}|$ is finite or countably infinite, there exists a RV $U$ supported on $\mathcal{U}$ such that $X$ and $U$ are independent, i.e., 
	\begin{align}\label{t1wi}
	I(U;X)= 0,
	\end{align}
	$Y$ is a deterministic function of $U$ and $X$, i.e., 
	\begin{align}\label{t2wi}
	H(Y|U,X)=0,
	\end{align}
	and 
	\begin{align}\label{prwi}
	|\mathcal{U}|\leq |\mathcal{X}|(|\mathcal{Y}|-1)+1.
	\end{align}
	Furthermore, if $X$ is a deterministic function of $Y$, we have
	\begin{align}\label{provewi}
	|\mathcal{U}|\leq |\mathcal{Y}|-|\mathcal{X}|+1.
	\end{align}
\end{lemma}  
As argued in \cite{kostala}, the proof of Lemma~\ref{FRL} is constructive and can be used to obtain the next lemma.
Next, we provide an extension of FRL %from \cite{kostala}, which shows that there exists a RV $U$ that satisfies \eqref{t1wi} and \eqref{t2wi}, and has bounded entropy. Lemma~\ref{aghabwi} has a constructive proof as well, that can 
that helps us find upper bound on $\mathbb{L}(P_{XY_1\cdot Y_N},T,C)$. %and $\mathbb{L}^{\epsilon}(P_{XY_1\cdot Y_N},T)$. 

\begin{lemma}\label{aghabwi} (\cite[Lemma~2]{kostala})
	For any pair of RVs $(X,Y)$ distributed according to $P_{XY}$ supported on alphabets $\mathcal{X}$ and $\mathcal{Y}$, respectively, where $|\mathcal{X}|$ is finite and $|\mathcal{Y}|$ is finite or countably infinite, there exists a RV $U$ such that it satisfies \eqref{t1wi}, \eqref{t2wi}, and
	\begin{align}
	H(U)\leq \sum_{x\in\mathcal{X}}H(Y|X=x).\label{wi}
	\end{align}
\end{lemma}
 Next, we present a summary of the achievable scheme proposed in \cite[Theorem~1]{maddah1}. We first consider a cache size $M\in\{\frac{N}{K},\frac{2N}{K},\ldots,N\}$ and define $p\triangleq\frac{MK}{N}$, which is an integer. In the placement phase, each file, e.g., $Y_n$, $n\in[N]$, is split into $\binom{K}{p}$ equal size subfiles and labeled as follows
\begin{align}
Y_n=(Y_{n,\Omega}:\Omega\subset[K],|\Omega|=p).
\end{align}
For all $n$, the server places $Y_{n,\Omega}$ in the cache of user $k$ if $k\in \Omega$. As argued in \cite{maddah1}, each user caches total of $N\binom{K-1}{p-1}\frac{F}{\binom{K}{p}}=MF$ bits, which satisfies the memory constraint with equality. In the delivery phase, %fix the request vector $(Y_{d_1},..,Y_{d_K})$ where user $i$ asks for $Y_{d_i}$. The 
the server sends the following message over the shared link
\begin{align}\label{cache1}
\mathcal{C}'\triangleq(C_{\gamma_1},\ldots,C_{\gamma_L}),
\end{align}
where $L=\binom{K}{p+1}$ and for any $i\in\{1,\ldots,L\}$, $\gamma_i$ is the $i$-th subset of $[K]$ with cardinality $|\gamma_i|=p+1$, furthermore,
\begin{align}\label{cache2}
C_{\gamma_i}\triangleq\oplus_{j\in \gamma_i} Y_{d_j,\gamma_i \backslash \{j\} },
\end{align}
where $\oplus$ denotes bitwise XOR function. Note that $Y_{d_j,\gamma_i \backslash \{j\} }$ is the subfile that is not cached by user $j$, but is requested by it. In other words, considering each subset of $[K]$ with cardinality $|\gamma_i|=p+1$, using the message $C_{\gamma_i}$, each user, e.g., user $j$, is able to decode the subfile $Y_{d_j,\gamma_i \backslash \{j\} }$ that is not cached by it. Considering all the messages in \eqref{cache1} user $i$ can decode file $Y_{d_i}$ completely using $\mathcal{C}'$ and its local cache content $Z_i$. Note that each subfile $C_{\gamma_i}$ has size $\frac{F}{\binom{K}{p}}$ bits. As pointed out in \cite{maddah1}, for other values of $M$ we can use the memory-sharing technique. For more details see \cite[Proof of Theorem~1]{maddah1}.

 Next, we recall the important results regarding upper and lower bounds on the minimum entropy coupling as obtained in \cite{kocaoglu2017entropic,compton2023minimum,shkel2023information}. Similar to \cite{shkel2023information}, for a given joint distribution $P_{XY}$ let the minimum entropy of functional representation of $(X,Y)$ be defined as
\begin{align}
H^*(P_{XY})=\!\!\!\!\!\inf_{\begin{array}{c} 
	\substack{H(Y|X,U)=0,\ I(X;U)=0}
	\end{array}}H(U).\label{minent}
\end{align}
\begin{remark}
	\normalfont
	By letting $\alpha=1$ in \cite[Definition 1]{shkel2023information}, it leads to the same problem in \eqref{minent}.
\end{remark}
\begin{remark}
	\normalfont
	As shown in \cite[Lemma 1]{shkel2023information}, the minimum entropy functional representation and the minimum entropy coupling are related functions. More specifically, $H^*(P_{XY})$ equals to the minimum entropy coupling of the set of PMFs $\{P_{Y|X=x_1},\ldots,P_{Y|X=x_n}\}$, where $\mathcal{X}=\{x_1,\ldots,x_n\}$.
\end{remark}
Let $\mathcal{G}_S$ be the output of the greedy entropy-based algorithm which is proposed in \cite[Section 3]{kocaoglu2017entropic}, i.e., $H^*(P_{XY})\leq H(\mathcal{G}_S)$. More specifically, the corresponding algorithm aims to solve \eqref{minent} but does not achieve the optimal solution in general. 

Next, we recall a result obtained in \cite{compton2023minimum} which shows that $\mathcal{G}_S$ is optimal within $\frac{\log e}{e}\approx 0.53$ bits for $|\mathcal{X}|=2$ and is optimal within $\frac{1+\log e}{2}\approx 1.22$ bits for $|\mathcal{X}|>2$.  
Let $U^*$ achieve the optimal solution of \eqref{minent}, i.e., $H(U^*)=H^*(P_{XY})$.
\begin{theorem}\cite[Th. 3.4, Th. 4.1, Th. 4.2]{compton2023minimum}
	Let $(X,Y)\sim P_{XY}$ and have finite alphabets. When $X$ is binary, we have
	\begin{align}
	H(Profile)&\leq H(U^*)\leq H(\mathcal{G}_S) \leq H(Profile)+\frac{\log e}{e}\nonumber \\ & \approx H(Profile)+0.53.
	\end{align}
	Moreover, for $|\mathcal{X}|>2$ we have
	\begin{align}
	H(Profile)&\leq H(U^*)\leq H(\mathcal{G}_S)\nonumber \\&\leq H(Profile)+\frac{1+\log e}{2}\nonumber \\ & \approx H(Profile)+1.22.
	\end{align}
	Here, $Profile$ corresponds to the profile method proposed in \cite[Section 3]{compton2023minimum}. 
\end{theorem}
Next, we present results on lower bounds on $H^*(P_{XY})$ obtained in a parallel work \cite{shkel2023information}. The lower bounds are obtained by using information spectrum and majorization concepts.
\begin{theorem}\cite[Corollary 2, Th. 2]{shkel2023information}
	Let $(X,Y)\sim P_{XY}$ and have finite alphabets. By letting $\alpha=1$ in \cite[Corollary 2, Th. 2]{shkel2023information}, we have
	\begin{align}
	H(\wedge_{x\in\mathcal{X}}P_{Y|x} )\leq H(Q^*)\leq H(U^*).
	\end{align}
	where $\wedge$ corresponds to the greatest lower bound with respect to majorization and $Q^*$ is defined in \cite[Lemma 3]{shkel2023information}.
\end{theorem}
\begin{remark}
	\normalfont
	In contrast with \cite{compton2023minimum}, the lower bounds in \cite{shkel2023information} are obtained considering \textit{R{\'e}nyi} entropy in \eqref{minent}. In this paper, we consider \textit{Shannon} entropy which is a special case of \textit{R{\'e}nyi} entropy.
\end{remark}
\begin{remark}
	\normalfont
	As argued in \cite[Remark 1]{shkel2023information}, for $\alpha=1$ the (largest) lower bounds obtained in \cite{compton2023minimum} and \cite{shkel2023information} match. Thus, using Theorem 1, for binary $X$ we have
	\begin{align}
	H(Q^*)\leq H(U^*)\leq H(\mathcal{G}_S) \leq H(Q^*)+\frac{\log e}{e},
	\end{align}
	and for $|\mathcal{X}|>2$,
	\begin{align}
	H(Q^*)&\leq H(U^*)\leq H(\mathcal{G}_S)\leq H(Q^*)+\frac{1+\log e}{2}.
	\end{align}
	Moreover, in some cases the lower bound $H(Q^*)$ is tight, e.g., see \cite[Example 2]{shkel2023information}.
\end{remark}
As outlined in \cite{shkel2023information}, the optimal distribution $Q^* = (q_1^*, q_2^*, \ldots)$ with components ordered such that $q_1^* \geq q_2^* \geq \cdots$, can be constructed using a greedy algorithm. Let $q_i = P(Q = q_i)$ denote the probability mass assigned to the $i$-th component. Consider the conditional probability matrix $P_{Y|X}$, where each column corresponds to a conditional distribution vector $P_{Y|X=x}$ for each $x \in \mathcal{X}$. Without loss of generality, assume that the entries in each column are sorted in descending order.
The greedy procedure begins by setting $q_1^*= \min_{x\in\mathcal{X}}\{\max_{y\in{\mathcal{Y}}} P_{Y|X}(y|x)\}$, which corresponds to selecting the smallest value in the first row of the matrix $P_{Y|X}$. This value is then subtracted from the entries in the first row, the matrix is updated accordingly, and each column is re-sorted to maintain descending order. The next value, $q_2^*$, is obtained by again selecting the smallest value in the updated first row. This iterative procedure continues until the sum of the selected values $q_i^*$ equals one, thereby constructing the full distribution $Q^*$. For an illustrative example of this method, see \cite[Example 1]{shkel2023information}.

Next, we recall the definition of the common information between $X$ and $Y$ using \cite{wyner}. For any pair of RVs $(X,Y)$ defined on discrete alphabets $\cal X$ and $\cal Y$, the common information between $X$ and $Y$ can be defined as follows
\begin{align}
C(X;Y)=\inf_{P_{W|XY}:X-W-Y} I(X,Y;W).
\end{align} 
As shown in \cite[Remark A]{wyner} we have
\begin{align}\label{hit}
I(X;Y)\leq C(X;Y)\leq \min\{H(X),H(Y)\}.
\end{align}
One simple observation is that when $H(X|Y)=0$ or $H(Y|X)=0$ we have $I(X;Y)= C(X;Y)$. This follows since when $H(X|Y)=0$ we have
\begin{align*}
I(X,Y;W)=I(Y;W)
\end{align*}
and $W$ can be chosen as $X$, hence $X-W-Y$ holds. However, these are not the only cases where we have $I(X;Y)= C(X;Y)$ \cite{ahlswede2006common}.
In the next example, we provide a more general case where the equality $I(X;Y)= C(X;Y)$ holds.
\begin{example}
	As argued in \cite{comm}, when it is possible to represent $X$ and $Y$ by $(X',V)$ and $(Y',V)$ where $X'$ and $Y'$ are conditionally independent given $V$, we have $I(X;Y)= C(X;Y)$. It is shown that this result is equivalent to the possibility of representing the probability matrix $P_{XY}$ as follows 
	\begin{align*}
	\begin{bmatrix}
	A_1 & 0 &  & 0\\
	0 & A_2 &  & 0\\
	&  & \cdot &\\
	0 & 0 &  & A_k	 
	\end{bmatrix}.
	\end{align*}
	Later follows since we have $p_{y,x}=p_{y'|v}p_{x'|v}p_{v}$. 
\end{example}

In the following we recall the definition of the common information using \cite{gacs1973common}. The G{\'a}cs-K{\"o}rner common information between two RVs $X$ and $Y$ is defined as the entropy of the common part between $X$ and $Y$, i.e., $C(X;Y)=H(U)$, where $U$ is the common part between $X$ and $Y$ \cite{gacs1973common}. The common part $U$ is defined based on the graphical representation of $X$ and $Y$, which is shown to be a deterministic function of only $X$ and only $Y$. Moreover, the G{\'a}cs-K{\"o}rner common information satisfies $C(X;Y)\leq I(X;Y)$ with equality if and only if the Markov chain $X-U-Y$ holds \cite{ahlswede2006common}. For more information about the G{\'a}cs-K{\"o}rner common information see \cite{gacs1973common} and \cite[Appendix B]{wang}. \\
In this paper, we can use both notions of common information defined in \cite{wyner} or \cite{gacs1973common}, since our focus is on scenarios where the common information and mutual information between $X$ and $Y$ are equal, i.e., $C(X;Y)=I(X;Y)$. In other words, our results hold for both Wyner and G{\'a}cs-K{\"o}rner notions of common information.

 \section{Main Results}\label{sec:resul}
 In this section, we derive upper bounds on $\mathbb{L}(P_{XY_1\cdot Y_N},T)$ defined in \eqref{main1wi}. For this, we utilize the two-part code construction, which has been used in \cite{kostala}. We first encode the private data $X$ using a one-time pad \cite[Lemma~1]{kostala2}, then encode the RV found by the achievable scheme in \cite[Theorem~1]{maddah1} by using the Functional Representation Lemma and blocks of $CF$ bits. %To do so, let us, first recall FRL.
To obtain an alternative achievability similarly we employ the two-part code construction. We first encode the private data $X$ using a one-time pad \cite[Lemma~1]{kostala2}, then encode the RV found by the achievable scheme in \cite[Theorem~1]{maddah1} by using the greedy entropy-based algorithm in \cite{kocaoglu2017entropic}. Finally, we propose an scheme using common information concept.
 \subsection{First design based on the FRL (Lemma 1):}\label{ty}
 In this part, we present our first achievable scheme which leads to upper bounds on \eqref{main1wi}. For simplicity let % $Y=(Y_1,..,Y_N)$ where $|\mathcal{Y}|$ is the size of database where $|\mathcal{Y}|\leq |\mathcal{Y}_1|\times..\times|\mathcal{Y}_N|$ and $\mathcal{Y}\subseteq \mathcal{Y}_1\times..\times\mathcal{Y}_N$. Furthermore, 
 $|\mathcal{C}'|$ be the cardinality of the codeword defined in \eqref{cache1} where $|\mathcal{C}'|= |\mathcal{C}_{\gamma_1}|\times\ldots\times|\mathcal{C}_{\gamma_L}|$. 
 \begin{theorem}\label{th1}
 	Let RVs $(X,Y)=(X,Y_1,\ldots,Y_N)$ be distributed according to $P_{XY_1\cdot Y_N}$ supported on alphabets $\mathcal{X}$ and $\mathcal{Y}$, where $|\mathcal{X}|$ is finite and $|\mathcal{Y}|$ is finite or countably infinite, and let the shared secret key size be $|\mathcal{X}|$, i.e., $T=|\mathcal{X}|$. Furthermore, let $M\in\{\frac{N}{K},\frac{2N}{K},\ldots,N\}$. Then, we have
 	\begin{align}
 	\mathbb{L}(P_{XY},|\mathcal{X}|)\leq \!\sum_{x\in\mathcal{X}}\!\!H(\mathcal{C}'|X=x)\!+\!1+\!\ceil{\log (|\mathcal{X}|)},\label{koonwi}
 	\end{align}
 	where $\mathcal{C}'$ is as defined in \eqref{cache1}, and if $|\mathcal{C}'|$ is finite we have
 	\begin{align}
 	\mathbb{L}(P_{XY},|\mathcal{X}|) \leq \ceil{\log\left(|\mathcal{X}|(|\mathcal{C}'|-1)+1\right)}+\ceil{\log (|\mathcal{X}|)}.\label{koon2wi}
 	\end{align}
 %	where $\mathcal{C}'$ is as defined in \eqref{cache1}.
 	Finally, if $X$ is a deterministic function of $\mathcal{C}'$, we have
 	\begin{align}\label{gohwi}
 	\mathbb{L}(P_{XY},|\mathcal{X}|) \leq \ceil{\log(\left(|\mathcal{C}'|-|\mathcal{X}|+1\right))}+\ceil{\log (|\mathcal{X}|)}.
 	\end{align}
 \end{theorem}
 \begin{sproof}
 	The complete proof is provided in Appendix~A.
 	In the placement phase, we use the same scheme as discussed before. In the delivery phase, we use the following strategy.
 	We use two-part code construction to achieve the upper bounds. As shown in Fig. \ref{achieve}, we first encode the private data $X$ using one-time pad coding \cite[Lemma~1]{kostala2}, which uses $\ceil{\log(|\mathcal{X}|)}$ bits. The rest follows since in the one-time pad coding, the RV added to $X$ is the shared key, which is of size $|\mathcal{X}|$, and as a result the output has uniform distribution.
 	Next, we produce $U$ based on FRL in Lemma~\ref{FRL} using $Y\leftarrow \mathcal{C}'$ and $X \leftarrow X$, where $\mathcal{C}'$, defined in \eqref{cache1}, is the response that the server sends over the shared link to satisfy the users$'$ demands \cite{maddah1}. Note that to produce such a $U$ we follow the construction used in Lemma \ref{aghabwi}. Thus, we have
 	\begin{align}
 	H(\mathcal{C}'|X,U)&=0,\label{kharkosde}\\
 	I(U;X)&=0,
 	\end{align}  
 	%Since we used the construction as in Lemma \ref{FRL}, 
 	and $U$ also satisfies \eqref{wi}. Thus, we obtain
 	\begin{align*}
 	\mathbb{L}(P_{XY},|\mathcal{X}|)\leq \sum_{x\in\mathcal{X}}H(\mathcal{C}'|X=x)+\!1+\!\ceil{\log (|\mathcal{X}|)}.
 	\end{align*}
 	To prove \eqref{koon2wi} we use the same coding scheme with the bound in \eqref{prwi}. If $X$ is a deterministic function of $\mathcal{C}'$, we can use the bound in \eqref{provewi} that leads to \eqref{gohwi}. Moreover, for the leakage constraint we note that the randomness of one-time-pad coding is independent of $X$ and the output of the FRL. %Let $Z$ be the compressed output of the EFRL and $\tilde{X}$ be the output of the one-time-pad coding.
 \begin{figure}[h]
 	\centering
 	\includegraphics[scale = .1]{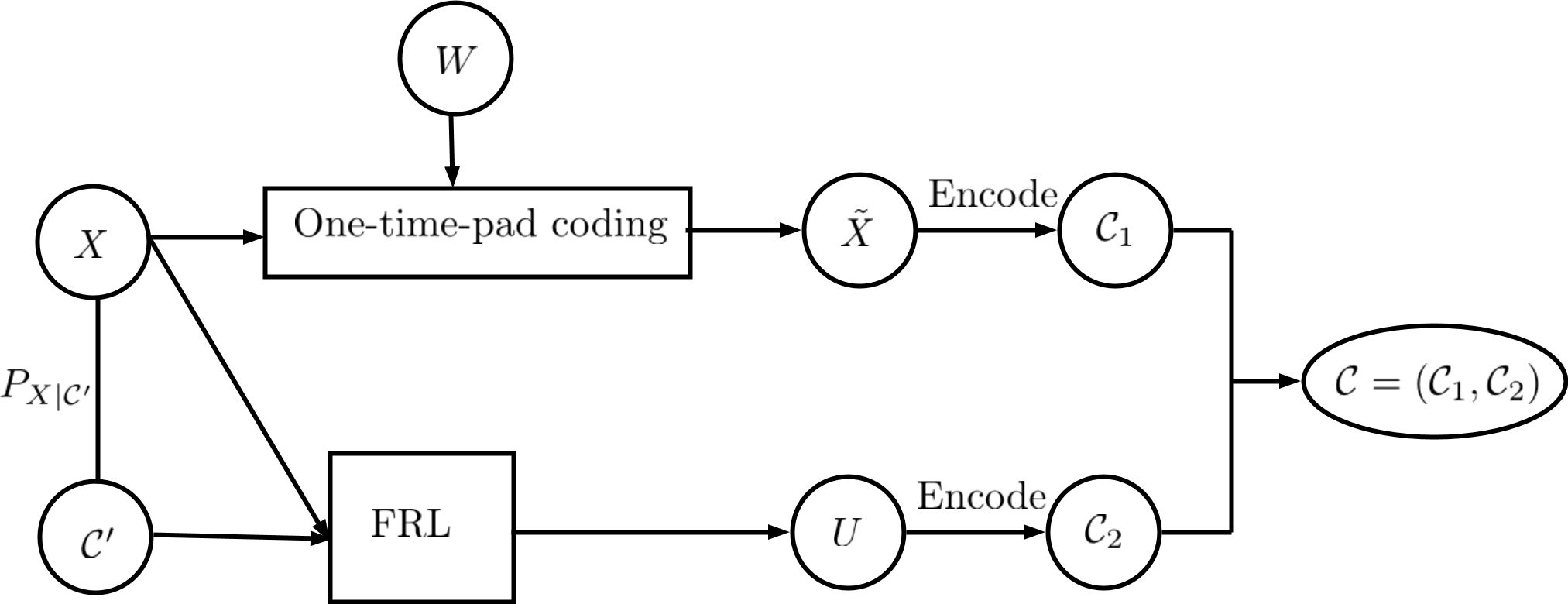}
 	\caption{Encoder design: illustration of the achievability scheme of Theorem \ref{th1}. Two-part code construction is used to produce the response of the server, $\mathcal{C}$. The server sends $\cal C$ over the channel, which is independent of $X$.}
 	%In this work an encoder wants to compress $Y$ which is correlated with $X$ under certain privacy leakage constraints and send it over a channel where an eavesdropper has access to the output of the encoder. The encoder and decoder have an advantage of using shared secret key.}
 	\label{achieve}
 \end{figure}
As shown in Fig. \ref{decode}, at user side, each user, e.g., user $i$, first decodes $X$ using one-time-pad decoder. Then, based on \eqref{kharkosde} it decodes $\mathcal{C}'$ using $U$ and $X$. Finally, it decodes $Y_{d_i}$ using local cache $Z_i$ and the response $\mathcal{C}'$.  
\begin{figure}[h]
	\centering
	\includegraphics[scale = .1]{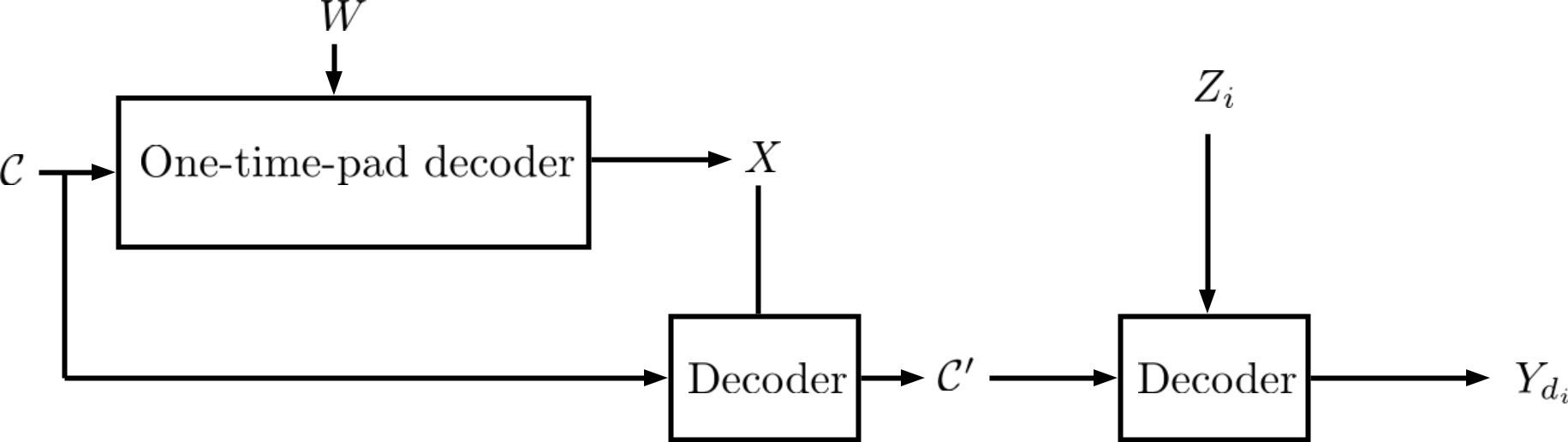}
	\caption{Illustration of the decoding process for the achievability scheme of Theorem \ref{th1}.}
	%In this work an encoder wants to compress $Y$ which is correlated with $X$ under certain privacy leakage constraints and send it over a channel where an eavesdropper has access to the output of the encoder. The encoder and decoder have an advantage of using shared secret key.}
	\label{decode}
\end{figure}
  \end{sproof}
\begin{remark}\label{re}
	\normalfont
	The upper bounds in Theorem~\ref{th1} also hold if the size of the database is countably infinite.  This result follows from the fact that Lemma~\ref{FRL} also applies to the case of a countably infinite $|\mathcal{Y}|$. %This helps us to use the achievable scheme when the size of database is large. 
\end{remark}
\begin{remark}
	\normalfont
	Although in Theorem~\ref{th1} we assume that $M\in\{\frac{N}{K},\frac{2N}{K},\ldots,N\}$, the results can be extended for other values of $M$ using the memory sharing technique of \cite[Theorem~1]{maddah1}.  
\end{remark}
 \begin{remark}
 	\normalfont
 	In this work, we assume that the privacy leakage constraint needs to be satisfied in the delivery phase and not in the placement phase. This assumption can be motivated since the placement phase occurs during the off-peak hours and we can assume that the adversary does not listen to the channel during that time, thus, sending coded contents of the database does not violate the privacy criterion. Furthermore, it is assumed that the adversary does not have access to the local cache contents. It only has access to the channel during the peak hours of the network, i.e., delivery phase.
 \end{remark}
\begin{remark}
	\normalfont
	Considering the scenarios in the presence of the adversary during the placement phase, the server can utilize the same strategy as used during the delivery phase. The server can fill the caches using the two-part code construction where we hide the information about $X$ by using the one-time pad coding.
\end{remark}
 %\begin{remark}
 %	In this work, 
 %\end{remark}
 Next we study a numerical example to better illustrate the achievable scheme in Theorem~\ref{th1}.
 \begin{example}
 	Let $F=N=K=2$ and $M=1$. Thus, $Y_1=(Y_1^1,Y_1^2)$ and $Y_2=(Y_2^1,Y_2^2)$, where $Y_i^j\in\{0,1\}$ for $i,j\in\{1,2\}$. We assume that $X$ is the pair of first bits of the database, i.e., $X=(Y_1^1,Y_2^1)$, $Y_1$ and $Y_2$ are independent and have the following distributions
 	\begin{align*}
 	P(Y_{1}^1 =Y_{1}^2=0)&=P(Y_{1}^1 =Y_{1}^2=1)=\frac{1}{16},\\
 	P(Y_{1}^1 =1,Y_{1}^2=0)&=P(Y_{1}^1 =0,Y_{1}^2=1)=\frac{7}{16},\\
 	P(Y_{2}^1 =Y_{2}^2=0)&=P(Y_{2}^1 =Y_{2}^2=1)=\frac{1}{10},\\
 	P(Y_{2}^1 =1,Y_{2}^2=0)&=P(Y_{2}^1 =0,Y_{2}^2=1)=\frac{2}{5},
 	\end{align*}
 	%where $\oplus$ denotes the XOR function. 
 	In this case, the marginal distributions can be calculated as $P(Y_1^1=1)=P(Y_1^2=1)=P(Y_2^1=1)=P(Y_2^2=1)=\frac{1}{2}$. In the placement phase, the server fills the first local cache by the first bits of the database, i.e., $Z_1=\{Y_1^1,Y_2^1\}$ and the second local cache by the second bits, i.e., $Z_2=\{Y_1^2,Y_2^2\}$. In the delivery phase, assume that users $1$ and $2$ request $Y_1$ and $Y_2$, respectively, i.e., $Y_{d_1}=Y_1$ and $Y_{d_2}=Y_2$. In this case, $\mathcal{C}'=Y_1^2\oplus Y_2^1$, where $\mathcal{C}'$ is the server's response without considering the privacy constraint. Thus, $|\mathcal{C}'|=2$. Next, we encode $X$ using $W$ as follows
 	\begin{align*}
 	\tilde{X}=X+W\ \text{mod}\ 4,
 	\end{align*}
 	where $W\perp X$ is a RV with uniform distribution over $\cal X$. To encode $\tilde{X}$ we use 2 bits. We then encode $\mathcal{C}'$ using Lemma \ref{FRL}. Let $U$ denote the output of the FRL which satisfies $H(U)\leq \sum_{x\in \mathcal{X}} H(\mathcal{C}'|X=x)$ based on the construction. Let $\mathcal{C}_1$ and $\mathcal{C}_2$ describe the encoded $\tilde{X}$ and $U$, respectively. The server sends $\mathcal{C}=(\mathcal{C}_1,\mathcal{C}_2)$ over the shared link. We have
 	\begin{align*}
 	&\sum_{x\in \mathcal{X}} H(\mathcal{C}'|X=x) = H(Y_1^2 \oplus Y_2^1| (Y_1^1,Y_2^1)=(0,0))\\&+H(Y_1^2 \oplus Y_2^1| (Y_1^1,Y_2^1)=(0,1))\\&+H(Y_1^2 \oplus Y_2^1| (Y_1^1,Y_2^1)=(1,0))\\&+H(Y_1^2 \oplus Y_2^1| (Y_1^1,Y_2^1)=(1,1))\\&=2\left(H(Y_1^2|Y_1^1=0)+H(Y_1^2|Y_1^1=1)\right)\\&=4h(\frac{1}{8})\\&=2.1743\ \text{bits}.
 	\end{align*}
 	For this particular demand vector, using \eqref{koonwi}, the average codelength is $5.1743$ bits and by using \eqref{koon2wi} $\ceil{\log(5)}+2=5$ bits are needed. For the request vector $(Y_{d_1},Y_{d_2})=(Y_1,Y_2)$, the average length of the code is $5$ bits to satisfy the zero leakage constraint. Thus, for $(Y_{d_1},Y_{d_2})=(Y_1,Y_2)$, we have
 	\begin{align*}
 	\mathbb{L}(P_{XY},4) \leq 5\ \text{bits}.
 	\end{align*}
 Both users first decode $X$ using $\tilde{X}$ and $W$, then decode $\mathcal{C}'=Y_1^2\oplus Y_2^1$ by using $X$ and $U$, since from FRL $\mathcal{C}'$ is a deterministic function of $U$ and $X$. User $1$ can decode $Y_1^2$ using $=Y_1^2\oplus Y_2^1$ and $Y_1^1$, which is available in the local cache $Z_1$, and user $2$ can decode $Y_2^1$ using $=Y_1^2\oplus Y_2^1$ and $Y_1^2$, which is in $Z_2$. Moreover, we choose $W$ to be independent of $X$ and $U$. As a result, $X$ and $(\tilde{X},U)$ become independent. Thus, $I(\mathcal{C};X)=0$, which means there is no leakage from $X$ to the adversary.
 Next, assume that in the delivery phase both users request $Y_1$, i.e., $Y_{d_1}=Y_{d_2}=Y_1$. In this case, $\mathcal{C}'=Y_1^2\oplus Y_1^1$ with $|\mathcal{C}'|=2$. We have
 \begin{align*}
 &\sum_{x\in \mathcal{X}} H(\mathcal{C}'|X=x) = H(Y_1^2 \oplus Y_1^1| (Y_1^1,Y_2^1)=(0,0))\\&+H(Y_1^2 \oplus Y_1^1| (Y_1^1,Y_2^1)=(0,1))\\&+H(Y_1^2 \oplus Y_1^1| (Y_1^1,Y_2^1)=(1,0))\\&+H(Y_1^2 \oplus Y_1^1| (Y_1^1,Y_2^1)=(1,1))\\&=2\left(H(Y_1^2 \oplus Y_1^1|Y_1^1=0)+H(Y_1^2 \oplus Y_1^1|Y_1^1=1)\right)\\&=2\left(H(Y_1^2|Y_1^1=0)+H(Y_1^2|Y_1^1=1)\right)\\&=4h(\frac{1}{8})=2.1743\ \text{bits}.
 \end{align*}
 Using \eqref{koonwi}, for $Y_{d_1}=Y_{d_2}=Y_1$, the average codelength is $5.1743$ bits and by using \eqref{koon2wi} $\ceil{\log(5)}+2=5$ bits are needed. 
 For the request vector $Y_{d_1}=Y_{d_2}=Y_1$, the average length of the code is $5$ bits to satisfy the zero leakage constraint.
 Next, let $Y_{d_1}=Y_{d_2}=Y_2$. In this case, $\mathcal{C}'=Y_2^1\oplus Y_2^2$. We have
 \begin{align*}
 &\sum_{x\in \mathcal{X}} H(\mathcal{C}'|X=x)\\&=2\left(H(Y_2^2\oplus Y_2^1 |Y_2^1=0)+H(Y_2^2\oplus Y_2^1|Y_2^1=1)\right)\\&=2\left(H(Y_2^2|Y_2^1=0)+H(Y_2^2|Y_2^1=1)\right)\\&=4h(\frac{1}{5})=2.8877\ \text{bits}.
 \end{align*}
 Using \eqref{koonwi}, for $Y_{d_1}=Y_{d_2}=Y_2$, the average codelength is $5.8877$ bits and by using \eqref{koon2wi} $\ceil{\log(5)}+2=5$ bits are needed. 
 For the request vector $Y_{d_1}=Y_{d_2}=Y_1$, the average length of the code is $5$ bits to satisfy the zero leakage constraint. Finally, let $Y_{d_1}=Y_2,\ Y_{d_2}=Y_1$. In this case, $\mathcal{C}'=Y_2^1\oplus Y_1^1$. In this case, since $\mathcal{C}'$ is a function of $X$, we have \eqref{koonwi} leads to zero and it is enough to only send $X$ using on-time pad coding. Thus, for the request vector $Y_{d_1}=Y_2,\ Y_{d_2}=Y_1$, the average length of the code is $2$ bits to satisfy the zero leakage constraint.
 As a results, since the code in \eqref{jojowi} is defined based on the worst case scenario, $5$ bits are needed to be sent over the shared link to satisfy the demands. The worst case scenarios correspond to the combinations $(Y_{d_1},Y_{d_2})=(Y_1,Y_2)$ and $Y_{d_1}=Y_{d_2}=Y_1$, respectively.  
 \end{example}
Next we provide lower bounds on \eqref{main1wi}.
\begin{theorem}\label{conv}
	 For any RVs $(X,Y_1,\ldots,Y_N)$ distributed according to $P_{XY_1\cdot Y_N}$, where $|\mathcal{X}|$ is finite and $|\mathcal{Y}|$ is finite or countably infinite and any shared key size $T\geq 1$ we have
	\begin{align}\label{lower}
	\mathbb{L}(P_{XY},T)\geq L_1(P_{XY}).
	\end{align}
	For and uncoded placement and $N$ independent $Y_i$, we have 
	\begin{align}\label{lower2}
	\mathbb{L}(P_{XY},T)\geq L_2(T,P_{XY}).
	\end{align}
	where 
	\begin{align}
	L_1(P_{XY}) &=\!\!\! \max_{t\in\{1,\ldots K\}}\!\!\left(\max_{x\in\mathcal{X}}H(Y_1,\ldots,Y_{t\floor{\frac{N}{t}}}|X=x)\right),\label{conv2}\\
	L_2(T,P_{XY}) &= \max_{t\in\{1,\ldots K\}}\!\!\left((t-\frac{t}{\floor{\frac{N}{t}}M})F\!-\!\log(T)\right).\label{conv1}
	\end{align}
\end{theorem}
\begin{proof}
	The proof is provided in Appendix A.
\end{proof}
In the next example we study the lower bound obtained in \eqref{conv2}.
\begin{example}
	Let $P_{XY_1Y_2}$ follow the same joint distribution as in Example 2. In this case, we have
	\begin{align*}
	L_1(P_{XY})&=\max_{y_1^1,y_2^1} H(Y_1,Y_2|Y_1^1=y_1^1,Y_2^1=y_2^1)\\
	&=\max_{y_1^1,y_2^1} H(Y_1^2|Y_1^1=y_1^1)+H(Y_2^2|Y_2^1=y_2^1)\\
	&=h(\frac{1}{8})+h(\frac{1}{5})\\&=1.2655 \ \text{bits}.
	\end{align*}
\end{example}
\subsection{Second design based on the minimum entropy coupling (Remark 7):}
 In the next result let $Q^*$ achieve the lower bound in Theorem 2 for the following problem 
\begin{align}
H^*(P_{X\mathcal{C}'})=\!\!\!\!\!\inf_{\begin{array}{c} 
	\substack{H(\mathcal{C}'|X,U)=0,\ I(X;U)=0}
	\end{array}}H(U),\label{minent2}
\end{align}
where in \eqref{minent}, $Y$ is substituted by $\mathcal{C}'$, i.e., $\mathcal{C}'\leftarrow Y$, and $\mathcal{C}'$ is as defined in \eqref{cache1}. Using Theorem 2 and Remark 6, $H^*(P_{X\mathcal{C}'})$ can be lower bounded by $H(Q^*)$ and upper bounded by $H(Q^*)+0.53$ when $|X|=2$ and by $H(Q^*)+1.22$ when $|X|>2$. For binary $X$ we have
\begin{align}\label{kiun}
H(Q^*)\leq H^*(P_{X\mathcal{C}'}) \leq H(Q^*)+\frac{\log e}{e},
\end{align}
and for $|X|>2$,
\begin{align}\label{kiun1}
H(Q^*)\leq H^*(P_{X\mathcal{C}'}) \leq H(Q^*)+\frac{1+\log e}{2}.
\end{align}
We emphasize that $Q^*$ that is used in \eqref{kiun} and \eqref{kiun1} is constructed using the greedy approach based on the matrix $P_{\mathcal{C}'|X}$. We use the same $Q^*$ for the following result.
\begin{theorem}\label{th12}
	Let RVs $(X,Y)=(X,Y_1,\ldots,Y_N)$ be distributed according to $P_{XY_1\cdot Y_N}$ supported on alphabets $\mathcal{X}$ and $\mathcal{Y}$, where $|\mathcal{X}|$ and $|\mathcal{Y}|$ are finite, and let the shared secret key size be $|\mathcal{X}|$, i.e., $T=|\mathcal{X}|$. Furthermore, let $M\in\{\frac{N}{K},\frac{2N}{K},\ldots,N\}$. Let $|X|=2$, we have
	\begin{align}
	\mathbb{L}(P_{XY},2)\leq \!H(Q^*)\!+\frac{\log e}{e}+\!2,\label{koonwie}
	\end{align}
	where $\mathcal{C}'$ is as defined in \eqref{cache1}. When $|X|>2$, we have
	\begin{align}
	\mathbb{L}(P_{XY},|\mathcal{X}|)\leq \!H(Q^*)\!+\frac{1+\log e}{2}+\!1+\!\ceil{\log (|\mathcal{X}|)},\label{koonwi2e}
	\end{align}
	%and if $|\mathcal{C}'|$ is finite we have
	% 	\begin{align}
	% 	\mathbb{L}(P_{XY},|\mathcal{X}|) \leq \ceil{\log\left(|\mathcal{X}|(|\mathcal{C}'|-1)+1\right)}+\ceil{\log (|\mathcal{X}|)}.\label{koon2wi}
	% 	\end{align}
	%	where $\mathcal{C}'$ is as defined in \eqref{cache1}.
	% 	Finally, if $X$ is a deterministic function of $\mathcal{C}'$, we have
	% 	\begin{align}\label{gohwi}
	% 	\mathbb{L}(P_{XY},|\mathcal{X}|) \leq \ceil{\log(\left(|\mathcal{C}'|-|\mathcal{X}|+1\right))}+\ceil{\log (|\mathcal{X}|)}.
	% 	\end{align}
\end{theorem}
\begin{sproof}
	The complete proof is provided in Appendix~A.
	The proof is similar to Section \ref{ty}. As shown in Fig. \ref{achieve2} the main difference is to use minimum entropy output of \eqref{minent2} instead of FRL that is used in two-part construction coding in Section \ref{ty}. As shown in Fig. \ref{decode2}, at user side, the decoding process is similar to Section \ref{ty}.
	 %Let $Z$ be the compressed output of the EFRL and $\tilde{X}$ be the output of the one-time-pad coding.
	\begin{figure}[h]
		\centering
		\includegraphics[scale = .1]{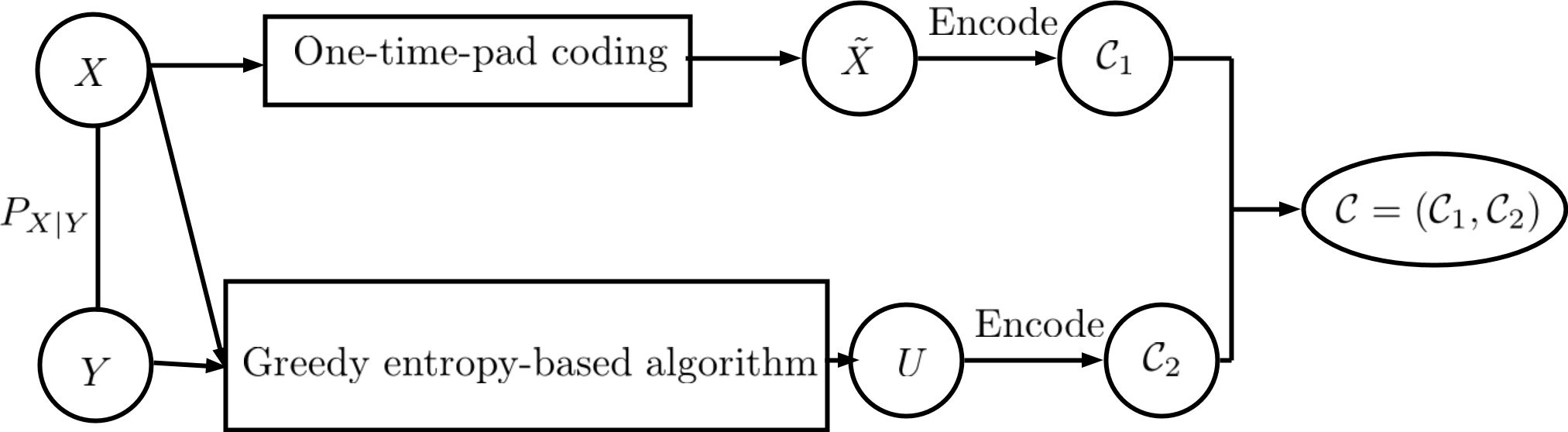}
		\caption{Encoder design: illustration of the achievability scheme of Theorem \ref{th12}. Two-part code construction is used to produce the response of the server, $\mathcal{C}$. The server sends $\cal C$ over the channel, which is independent of $X$.}
		\label{achieve2}
	\end{figure}  
	\begin{figure}[h]
		\centering
		\includegraphics[scale = .1]{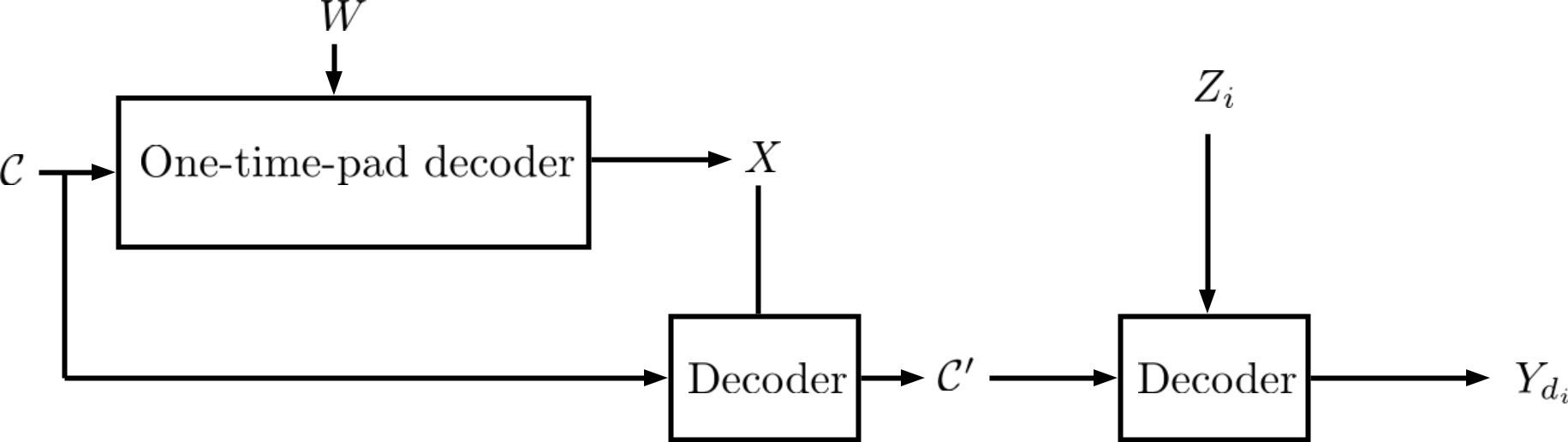}
		\caption{Illustration of the decoding process for the achievability scheme of Theorem \ref{th12}.}
		\label{decode2}
	\end{figure}
\end{sproof}
\begin{remark}
	\normalfont
	As we mentioned earlier, the main difference between the present scheme and the achievability in Section \ref{ty} is to use greedy entropy-based algorithm which aims to minimize the output of FRL and is optimal within a constant gap.  
\end{remark}
\begin{remark}
	\normalfont
	The complexity of the algorithm to find $Q^*$ in Theorem \ref{th12} is linear in $|\mathcal{C}'|\times|\mathcal{X}|$. This can be shown by using \cite[Lemma 3]{shkel2023information}.
\end{remark}
%\begin{remark}
%	\normalfont
%	Although in Theorem~\ref{th12} we assume that $M\in\{\frac{N}{K},\frac{2N}{K},\ldots,N\}$, the results can be extended for other values of $M$ using the memory sharing technique of \cite[Theorem~1]{maddah1}.  
%\end{remark}
\begin{remark}
	\normalfont
	Similar to Section \ref{ty}, we assume that the privacy leakage constraint needs to be fulfilled in the delivery phase. This assumption can be motivated since the placement phase occurs during the off-peak hours and we can assume that the adversary does not listen to the channel during that time. Considering the scenarios in the presence of the adversary during the placement phase, the server can employ the same strategy as used during the delivery phase. The server can fill the caches using the two-part code construction coding. %where we hide the information about $X$ by using the one-time pad coding.
\end{remark}
Next we study Example 2 to better illustrate the achievable scheme in Theorem~\ref{th12} and compare it with Theorem~\ref{th1}.
\begin{example}
	Let $P_{XY_1Y_2}$ follow the same joint distribution as in Example 2.
	%where $\oplus$ denotes the XOR function. 
	In this case, the marginal distributions can be calculated as $P(Y_1^1=1)=P(Y_1^2=1)=P(Y_2^1=1)=P(Y_2^2=1)=\frac{1}{2}$. In the placement phase, the server fills the first local cache by the first bits of the database, i.e., $Z_1=\{Y_1^1,Y_2^1\}$ and the second local cache by the second bits, i.e., $Z_2=\{Y_1^2,Y_2^2\}$. In the delivery phase, assume that users $1$ and $2$ request $Y_1$ and $Y_2$, respectively, i.e., $Y_{d_1}=Y_1$ and $Y_{d_2}=Y_2$. In this case, $\mathcal{C}'=Y_1^2\oplus Y_2^1$, where $\mathcal{C}'$ is the server's response without considering the privacy constraint. Thus, $|\mathcal{C}'|=2$ and 
	\begin{align*}
	P_{\mathcal{C}'|X}=\begin{bmatrix}
	\frac{7}{8} & \frac{1}{8} & \frac{1}{8} & \frac{7}{8}\\
	\frac{1}{8} & \frac{7}{8} & \frac{7}{8} & \frac{1}{8}
	\end{bmatrix}
	\end{align*}
	Moreover, $Q^*$ has the following distribution $P_{Q^*}=[\frac{7}{8}, \frac{1}{8}]$, hence, $H(Q^*)=h(1/8)=0.5436$.
	Next, we encode $X$ using $W$ as follows
	\begin{align*}
	\tilde{X}=X+W\ \text{mod}\ 4,
	\end{align*}
	where $W\perp X$ is a RV with uniform distribution over $\cal X$. To encode $\tilde{X}$ we use 2 bits. We then encode $\mathcal{C}'$ using greedy entropy-based algorithm. Let $U$ denote the output of the algorithm which satisfies \eqref{kiun3}. Let $\mathcal{C}_1$ and $\mathcal{C}_2$ describe the encoded $\tilde{X}$ and $U$, respectively. The server sends $\mathcal{C}=(\mathcal{C}_1,\mathcal{C}_2)$ over the shared link.
	% 	\begin{align*}
	% 	&\sum_{x\in \mathcal{X}} H(\mathcal{C}'|X=x) = H(Y_1^2 \oplus Y_2^1| (Y_1^1,Y_2^1)=(0,0))\\&+H(Y_1^2 \oplus Y_2^1| (Y_1^1,Y_2^1)=(0,1))\\&+H(Y_1^2 \oplus Y_2^1| (Y_1^1,Y_2^1)=(1,0))\\&+H(Y_1^2 \oplus Y_2^1| (Y_1^1,Y_2^1)=(1,1))\\&=2\left(H(Y_1^2|Y_1^1=0)+H(Y_1^2|Y_1^1=1)\right)\\&=4h(\frac{1}{8})=2.1743\ \text{bits}.
	% 	\end{align*}
	For this particular demand vector, using \eqref{koonwi2e}, the average codelength is $4.7636$ bits. For the request vector $(Y_{d_1},Y_{d_2})=(Y_1,Y_2)$, the average length of the code is $4.7636$ bits to satisfy the zero leakage constraint. Thus, for $(Y_{d_1},Y_{d_2})=(Y_1,Y_2)$, we have
	\begin{align*}
	\mathbb{L}(P_{XY},4) \leq 4.7636\ \text{bits}.
	\end{align*}
	Using Example 2, for this particular demand vector we need $5$ bits.
	Both users first decode $X$ using $\tilde{X}$ and $W$, then decode $\mathcal{C}'=Y_1^2\oplus Y_2^1$ by using $X$ and $U$, since from FRL $\mathcal{C}'$ is a deterministic function of $U$ and $X$. User $1$ can decode $Y_1^2$ using $=Y_1^2\oplus Y_2^1$ and $Y_1^1$, which is available in the local cache $Z_1$, and user $2$ can decode $Y_2^1$ using $=Y_1^2\oplus Y_2^1$ and $Y_1^2$, which is in $Z_2$. Moreover, we choose $W$ to be independent of $X$ and $U$. As a result, $X$ and $(\tilde{X},U)$ become independent. Thus, $I(\mathcal{C};X)=0$, which means there is no leakage from $X$ to the adversary. %By checking other demand vectors we conclude that the achievable scheme derived in \cite{amircache} is strengthened.
	Next, assume that in the delivery phase both users request $Y_1$, i.e., $Y_{d_1}=Y_{d_2}=Y_1$. In this case, $\mathcal{C}'=Y_1^2\oplus Y_1^1$ with $|\mathcal{C}'|=2$. Using the same arguments we need $4.7636$ bits.
	% \begin{align*}
	% &\sum_{x\in \mathcal{X}} H(\mathcal{C}'|X=x) = H(Y_1^2 \oplus Y_1^1| (Y_1^1,Y_2^1)=(0,0))\\&+H(Y_1^2 \oplus Y_1^1| (Y_1^1,Y_2^1)=(0,1))\\&+H(Y_1^2 \oplus Y_1^1| (Y_1^1,Y_2^1)=(1,0))\\&+H(Y_1^2 \oplus Y_1^1| (Y_1^1,Y_2^1)=(1,1))\\&=2\left(H(Y_1^2 \oplus Y_1^1|Y_1^1=0)+H(Y_1^2 \oplus Y_1^1|Y_1^1=1)\right)\\&=2\left(H(Y_1^2|Y_1^1=0)+H(Y_1^2|Y_1^1=1)\right)\\&=4h(\frac{1}{8})=2.1743\ \text{bits}.
	% \end{align*}
	% Using \eqref{koonwi}, for $Y_{d_1}=Y_{d_2}=Y_1$, the average codelength is $4.1743$ bits and by using \eqref{koon2wi} $\ceil{\log(5)}+2=5$ bits are needed. 
	% For the request vector $Y_{d_1}=Y_{d_2}=Y_1$, the average length of the code is $4.1743$ bits to satisfy the zero leakage constraint.
	Next, let $Y_{d_1}=Y_{d_2}=Y_2$. In this case, $\mathcal{C}'=Y_2^1\oplus Y_2^2$. In this case, $H(Q^*)=h(1/5)=0.7219$ and we need $4.9419$ bits. 
	Finally, let $Y_{d_1}=Y_2,\ Y_{d_2}=Y_1$. In this case, $\mathcal{C}'=Y_2^1\oplus Y_1^1$. Since $\mathcal{C}'$ is a function of $X$ it is enough to only send $X$ using on-time pad coding. Thus, for the request vector $Y_{d_1}=Y_2,\ Y_{d_2}=Y_1$, the average length of the code is $2$ bits to satisfy the zero leakage constraint.
	We conclude that in all cases we need less bits to send compared to Example 2, since by using Example 2 we need 5 bits on average to send over the channel.   
	% \begin{align*}
	% &\sum_{x\in \mathcal{X}} H(\mathcal{C}'|X=x)\\&=2\left(H(Y_2^2\oplus Y_2^1 |Y_2^1=0)+H(Y_2^2\oplus Y_2^1|Y_2^1=1)\right)\\&=2\left(H(Y_2^2|Y_2^1=0)+H(Y_2^2|Y_2^1=1)\right)\\&=4h(\frac{1}{5})=2.8877\ \text{bits}.
	% \end{align*}
	% Using \eqref{koonwi}, for $Y_{d_1}=Y_{d_2}=Y_2$, the average codelength is $4.8877$ bits and by using \eqref{koon2wi} $\ceil{\log(5)}+2=5$ bits are needed. 
	% For the request vector $Y_{d_1}=Y_{d_2}=Y_1$, the average length of the code is $4.8877$ bits to satisfy the zero leakage constraint. Finally, let $Y_{d_1}=Y_2,\ Y_{d_2}=Y_1$. In this case, $\mathcal{C}'=Y_2^1\oplus Y_1^1$. In this case, since $\mathcal{C}'$ is a function of $X$, we have \eqref{koonwi} leads to zero and it is enough to only send $X$ using on-time pad coding. Thus, for the request vector $Y_{d_1}=Y_2,\ Y_{d_2}=Y_1$, the average length of the code is $2$ bits to satisfy the zero leakage constraint.
	% As a results, since the code in \eqref{jojowi} is defined based on the worst case scenario, $4.8877$ bits are needed to be sent over the shared link to satisfy the demands. The worst case scenarios correspond to the combinations $(Y_{d_1},Y_{d_2})=(Y_1,Y_2)$ and $Y_{d_1}=Y_{d_2}=Y_1$, respectively.  
\end{example}
\subsection{Third design in a special case: improving the bounds using the common information concept}
In this section, we improve the bounds obtained in Theorem \ref{th1} considering a special case. To do so, let us recall the privacy mechanism design problems considered in \cite{shah} with zero leakage as follows
\begin{align}
g_{0}(P_{XY})&=\max_{\begin{array}{c} 
	\substack{P_{U|Y}:X-Y-U\\ \ I(U;X)=0,}
	\end{array}}I(Y;U),\label{maing}\\
h_{0}(P_{XY})&=\max_{\begin{array}{c} 
	\substack{P_{U|Y,X}: I(U;X)=0,}
	\end{array}}I(Y;U).\label{mainh}
\end{align} 
Finally, we define a set of joint distributions $\hat{\mathcal{P}}_{XY}$ as follows
\begin{align}
\hat{\mathcal{P}}_{XY}\triangleq \{P_{XY}:g_{0}(P_{XY})=h_{0}(P_{XY})\}.
\end{align}
As outlined in \cite[Lemma 1]{zero}, a sufficient condition to have $g_{0}(P_{XY})=h_{0}(P_{XY})$ is to have $C(X;Y)=I(X;Y)$, where $C(X,Y)$ denotes the common information between $X$ and $Y$, where common information corresponds to the Wyner \cite{wyner} or G{\'a}cs-K{\"o}rner \cite{gacs1973common} notions of common information. Moreover, a sufficient condition for $C(X;Y)=I(X;Y)$ is to let $X$ be a deterministic function of $Y$ or $Y$ be a deterministic function of $X$. In both cases, $$C(X;Y)=I(X;Y),$$ and $$g_{0}(P_{XY})=h_{0}(P_{XY}).$$ For more detail see \cite[Proposition 6]{shah}. Moreover, in \cite[Lemma 2]{zero}, properties of the optimizers for $g_{0}(P_{XY})$ and $h_{0}(P_{XY})$ are obtained considering $P_{XY}\in \hat{\mathcal{P}}_{XY}$. It has been shown that the optimizer $U^*$ achieving $g_{0}(P_{XY})=h_{0}(P_{XY})$ satisfies
\begin{align}
H(Y|U^*,X)=0,\label{2}\\
I(X;U^*|Y)=0,\label{3}\\
I(X;U^*)=0.\label{4}
\end{align} 
Next, we recall the definitions of a set $\mathcal{U}^1(P_{XY})$ and a function $\mathcal{K}(P_{XY})$ in \cite{zero} as follows. 
\begin{align}
\mathcal{U}^1(P_{XY})&\triangleq \{U: U\ \text{satisfies \eqref{2}, \eqref{3}, \eqref{4}}\}\\
\mathcal{K}(P_{XY})&\triangleq \min_{U\in \mathcal{U}^1(P_{XY})} H(U).
\end{align} 
Noting that the function $\mathcal{K}(P_{XY})$ finds the minimum entropy of all optimizers satisfying $g_0(P_{XY})=h_0(P_{XY})$. A simple bound on $\mathcal{K}(P_{XY})$ has been obtained in \cite[Lemma 4]{zero}. Next, we define
\begin{align}
A_{XY}&\triangleq \begin{bmatrix}
&P_{y_1}-P_{y_1|x_1} &\ldots & P_{y_{|\mathcal{Y}|}}-P_{y_{q}|x_1}\\
&\cdot &\ldots &\cdot\\
&P_{y_1}-P_{y_1|x_{t}} &\ldots & P_{y_{q}}-P_{y_{q}|x_{t}}
\end{bmatrix}\!\in\! \mathbb{R}^{t\times q},\\
b_{XY}&\triangleq \begin{bmatrix}
H(Y|x_1)-H(Y|X) \\
\cdot \\
H(Y|x_t)-H(Y|X)
\end{bmatrix}\in\mathbb{R}^{t},\ \bm{a}\triangleq \begin{bmatrix}
a_1 \\
\cdot \\
a_q
\end{bmatrix}\in\mathbb{R}^{q}.
\end{align}
where $t=|\mathcal{X}|$ and $q=|\mathcal{Y}|$. Noting that in \cite[Theorem 1]{zero}, bounds on $\mathcal{K}(P_{XY})$ and entropy of any $U\in\mathcal{U}^1$ have been obtained. Specifically, when $\text{rank}(A_{XY})=|\mathcal{Y}|$, the exact value of $\mathcal{K}(P_{XY})$ is obtained by solving simple linear equations in \cite[eq. (26)]{zero}. 
We emphasize that by using \cite{borz}, $g_0(P_{XY})$ can be obtained by solving a linear program in which the size of the matrix in the system of linear equations is at most $|\mathcal{Y}|\times\binom{|\mathcal{Y}|}{\text{rank}(P_{X|Y})}$ with at most $\binom{|\mathcal{Y}|}{\text{rank}(P_{X|Y})}$ variables. By solving the linear program as proposed in \cite{borz} we can find the exact value of $\mathcal{K}(P_{XY})$ and the joint distribution $P_{U|YX}$ that achieves it. The complexity of the linear program in \cite{borz} can grow faster than exponential functions with respect to $|\mathcal{Y}|$, however the complexity of the proposed method in \cite{zero} grows linearly with $|\mathcal{Y}|$. Thus, our proposed upper bound has less complexity compared to the solution in \cite{borz}.
The bounds on $\mathcal{K}(P_{XY})$ help us to obtain the next result. Next, we improve the bounds obtained in Theorem \ref{th1}.
\begin{theorem}\label{loo}
	Let RVs $(X,Y)=(X,Y_1,\ldots,Y_N)$ be distributed according to $P_{XY_1\cdot Y_N}$ supported on alphabets $\mathcal{X}$ and $\mathcal{Y}$, where $|\mathcal{X}|$ and $|\mathcal{Y}|$ are finite, and let the shared secret key size be $|\mathcal{X}|$, i.e., $T=|\mathcal{X}|$. Furthermore, let $M\in\{\frac{N}{K},\frac{2N}{K},\ldots,N\}$. Let $P_{X\mathcal{C}'}\in\hat{\mathcal{P}}_{X\mathcal{C}'}$ and let $q=|\mathcal{C}'|$ and $\beta=\log(\text{null}(P_{X|\mathcal{C}'})+1)$, where $\mathcal{C}'$ is defined in \eqref{cache1}. Then, we have 
	\begin{align}
	&\mathbb{L}(P_{X\mathcal{C}'},|\mathcal{X}|)%\leq \min_{\begin{array}{c} 
	%\substack{P_{U|Y,X}: I(U;X)=0,\\ H(Y|X,U)=0}
	%\end{array}}H(U)+1+\ceil{\log(|\mathcal{X}|)}\\&
	\leq \mathcal{K}(P_{X\mathcal{C}'})+1+\ceil{\log(|\mathcal{X}|)} \label{log}\\
	&\leq H(\mathcal{C}'|X)\!\!+\!\!\!\!\!\!\!\!\!\!\!\!\!\!\!\!\!\!\!\!\max_{\begin{array}{c}
		\substack{a_i:A_{XY}\bm{a}=b_{XY},\bm{a}\geq 0,\\
			\sum_{i=1}^{q}\! P_{c'_i}a_i\leq \beta-H(\mathcal{C}'|X)} \end{array} }\!\!\sum_{i=1}^{q} \!\!P_{c'_i}a_i\!+\!1\!+\!\ceil{\log(|\mathcal{X}|)}\label{ass}\\&\leq \beta+1+\!\ceil{\log(|\mathcal{X}|)},\label{mass}
	\end{align}
	where $c'_i$ is the $i$-th element (alphabet) of $\mathcal{C}'$.
	Moreover, we have
	\begin{align}
	&\mathbb{L}(P_{XY},2)\leq \!H(Q^*)\!+\frac{\log e}{e}+\!2,\label{koonwi11}\\
	&\mathbb{L}(P_{XY},|\mathcal{X}|)\leq \!H(Q^*)\!+\frac{1+\log e}{2}+\!1+\!\ceil{\log (|\mathcal{X}|)}\label{koonwi22},
	\end{align}
	where $Q^*$ is defined in Theorem \ref{th1}.
	%and if $X=f(Y)$,
	%\begin{align}\label{govad}
	%\mathbb{L}(P_{XY},|\mathcal{X}|)\leq \ceil{\log(|\mathcal{Y}|\!-\!|\mathcal{X}|\!+\!1)\!}\!+\!\ceil{\log(|\mathcal{X}|)}.
	%\end{align} 
	Finally, for any $P_{X\mathcal{C}'}$ (not necessarily $P_{X\mathcal{C}'}\in\hat{\mathcal{P}}_{X\mathcal{C}'}$) with $|\mathcal{C}'|\leq |\mathcal{X}|$ we have
	\begin{align}
	\mathbb{L}(P_{X\mathcal{C}'},|\mathcal{C}'|)\leq \ceil{\log{|\mathcal{C}'|}}.\label{kos3} 
	\end{align}
\end{theorem}
\begin{sproof}
	The complete proof is provided in Appendix A.
	The proof is based on two-part construction coding and is similar to Theorem \ref{th1} and \cite[Theorem 2]{zero}. As shown in Fig. \ref{kesh11}, to achieve \eqref{log}, we use the solution to $h_0(P_{X\mathcal{C}'})=g_0(P_{X\mathcal{C}'})$ instead of the greedy entropy-based algorithm. 
\end{sproof}
\begin{remark}
	Clearly, the upper bound obtained in \eqref{kos3} improves the bounds in Theorem \ref{th1}. Since, when $|\mathcal{C}'|\leq |\mathcal{X}|$ we have
	\begin{align}
	\ceil{\log{|\mathcal{C}'|}}\leq \!H(Q^*)\!+\frac{1+\log e}{2}+\!1+\!\ceil{\log (|\mathcal{X}|)}. 
	\end{align} 
\end{remark}
\begin{figure}[]
	\centering
	\includegraphics[scale = .12]{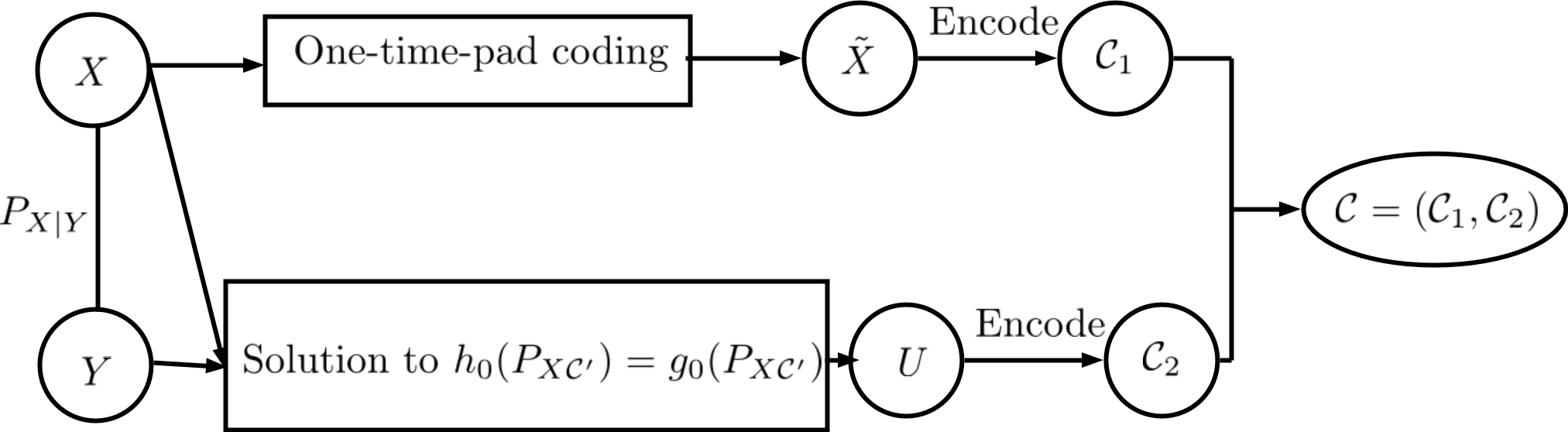}
	\caption{In this work, we use two-part construction coding strategy to send codewords over the channels. We hide the information of $X$ using one-time-pad coding and we then use the solution of $h_0(P_{X\mathcal{C}'})=g_0(P_{X\mathcal{C}'})$ to construct $U$.} %As discussed in Theorem~\ref{the1}, at receiver side using $\mathcal{C}$ and the shared key we decode $X$ and using $X$, the shared key and $\mathcal{C}$ we decode $Y$. The leakage from $X$ to $\mathcal{C}$ equals $\epsilon$.}
	\label{kesh11}
\end{figure}
Next, we provide a numerical example that shows \eqref{log} improves \eqref{koonwie}.
\begin{example}
	Let \begin{align*}
P_{X|\mathcal{C}'}=\begin{bmatrix}
1 &1 &1 &0 &0 &0\\
0 &0 &0 &1 &1 &1
\end{bmatrix}
	\end{align*} and $$P_{\mathcal{C}'}=[\frac{1}{8},\frac{2}{8},\frac{3}{8},\frac{1}{8},\frac{1}{16},\frac{1}{16} ].$$ Clearly, in this case $X$ is a deterministic function of $\mathcal{C}'$. Using the linear program proposed in \cite{borz}, we obtain a solution as \begin{align*}P_{\mathcal{C}'|u_1}&=[0.75,0,0,0.25,0,0],\\ P_{\mathcal{C}'|u_2}&=[0,0.75,0,0.25,0,0],\\ P_{\mathcal{C}'|u_3}&=[0,0,0.75,0,0.25,0],\\ P_{\mathcal{C}'|u_4}&=[0,0,0.75,0,0,0.25],\end{align*} and $$P_U=[\frac{1}{6},\frac{1}{3},\frac{1}{4},\frac{1}{4} ],$$ which results $H(U)=1.9591$ bits. We have $H(U)=\mathcal{K}(P_{XY})\leq 1.9591$. Moreover, we have
	\begin{align*}
	P_{\mathcal{C}'|X}=\begin{bmatrix}
	\frac{1}{6} &0\\ \frac{1}{3} &0\\ \frac{1}{2} &0\\ 0 &\frac{1}{2}\\0 &\frac{1}{4}\\0 &\frac{1}{4}
	\end{bmatrix}.
	\end{align*}
	Using the greedy search algorithm we have $$P_{Q*}=[\frac{1}{2}\ \frac{1}{4}\ \frac{1}{6}\ \frac{1}{12}],$$ hence, $H(Q^*)=1.7296$. Thus, 
	\begin{align*}
	H(Q^*)+\frac{\log e}{e}&=2.2596\\&\geq \mathcal{K}(P_{XY})\\&=1.9591.
	\end{align*}
\end{example} 
\section{Application: Cache-Aided networks with limited encoder buffer}\label{st}
\begin{figure}[]
	\centering
	\includegraphics[scale = .12]{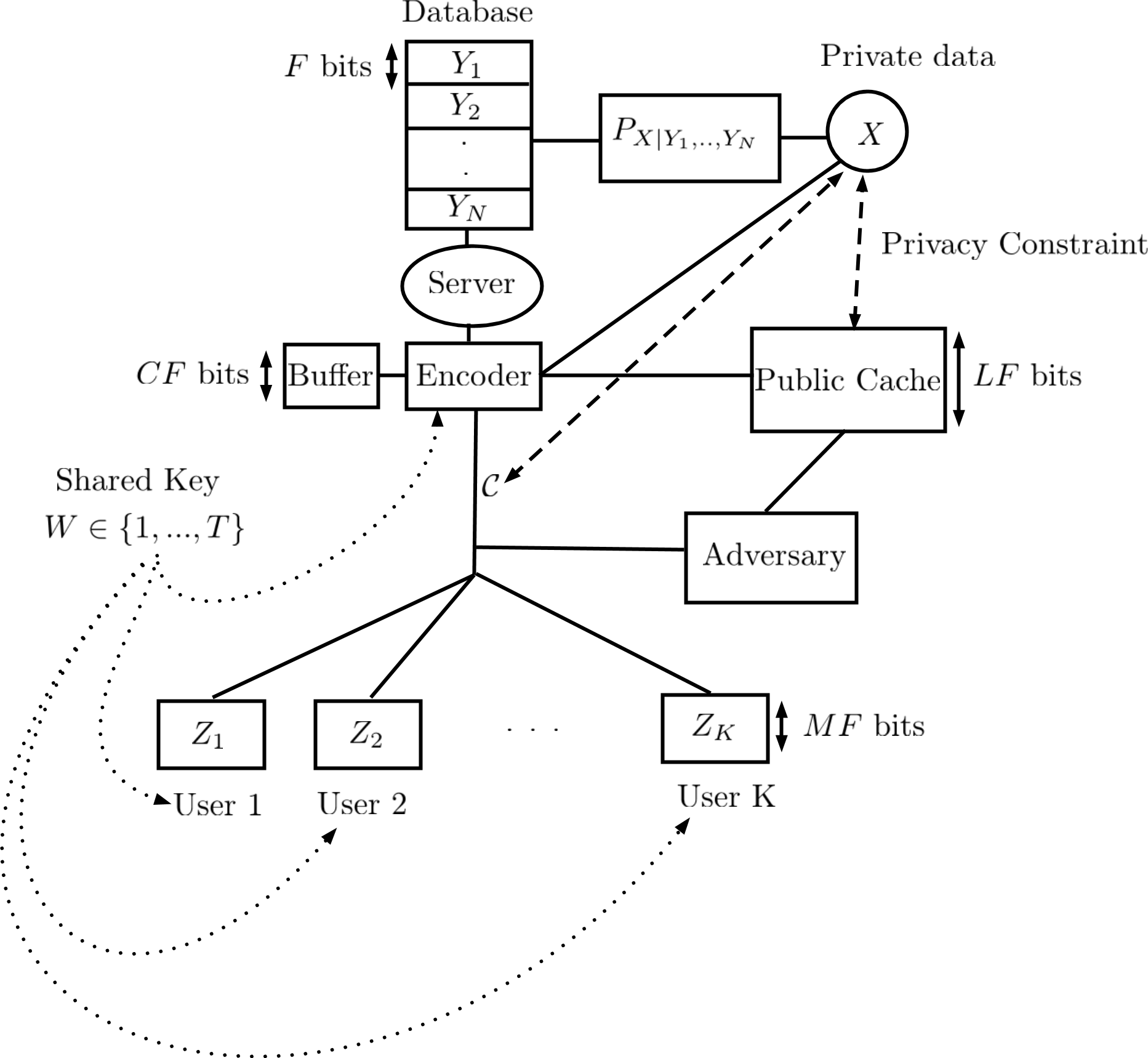}
	\caption{In this work, a server wants to send a response over a shared link to satisfy users' demands, but since the database is correlated with the private data existing schemes are not applicable. In this model, the server sends the response to an encoder and the encoder sequentially encodes the message $\mathcal{C}$ and transmits it.}
	\label{application}
\end{figure}
In this section, we present an application that considers an encoder equipped by a local buffer with limited size.
We consider the scenario depicted in Fig. \ref{application}, in which a server has access to a database containing $N$ files, denoted by $Y_1, \dots, Y_N$. Each file, of size $F$ bits, is generated according to the joint distribution $P_{X Y_1 \dots Y_N}$, where $X$ represents a private latent variable. The server communicates with an encoder equipped with a local buffer of capacity $CF$ bits and also has access to a public cache of size $LF$ bits. The encoder is connected to $K$ users via a shared link, with user $i$ having access to a local cache of size $MF$ bits. Additionally, the encoder and the users share a common key $W$ with size denoted by $T$, i.e., $W\in\{1,\ldots,T\}$.

The system operates in two distinct phases: the placement phase and the delivery phase. During the placement phase, the server and encoder fill the local caches using data from the database. Let $Z_k$ denote the cached content at user $k$, for $k \in [K]$, following this phase. In the subsequent delivery phase, users communicate their demands to the encoder and server, where $d_k \in [N]$ denotes the demand from user $k$. To fulfill these demands, the server transmits a response $\mathcal{C}' = (\mathcal{C}'_1, \mathcal{C}'_2, \dots)$ to the encoder, structured into blocks, each of size $CF$ bits. The encoder then constructs a message $\mathcal{C}_i$ from the $i$-th block $\mathcal{C}'_i$ of the server's response and the shared key, and transmits it to users via the shared link. The encoder can also store message $\mathcal{C}_i$ in the public cache, which is inaccessible to the users. This public cache allows the encoder to use previously stored messages $(\mathcal{C}1, \dots, \mathcal{C}{i-1})$ when designing subsequent messages $\mathcal{C}_i$.

We assume an adversary can observe both the shared communication link and the public cache, attempting to extract information about the private latent variable $X$ from the transmitted messages $\mathcal{C} = (\mathcal{C}_1, \mathcal{C}_2, \dots)$ and the content of the public cache $P$. Importantly, the adversary does not have access to the users' caches or the shared key. Since the files in the database are correlated with $X$, conventional coded caching and delivery schemes, such as those introduced in \cite{maddah1}, do not satisfy the necessary privacy constraints. The objective in this cache-aided private delivery scenario is to design the response $\mathcal{C}$ with the minimal average length possible, subject to both the privacy constraint and the users' zero-error decoding requirements. Additionally, the public cache content $P$ itself must satisfy the specified privacy constraint. We address the scenario involving worst-case demand patterns $d = (d_1, \dots, d_K)$ to build the message $\mathcal{C}$, with expectation taken over the randomness inherent in the database.

The output of the server $\mathcal{C}'(Y,d_1,\ldots,d_K)$ is the codeword the server sends to the encoder by using blocks of $CF$ bits to satisfy the demands of users $(Y_{d_1},\ldots,Y_{d_K})$. Let $\mathcal{C}'_i$ be the $i$-th block of the codeword $\mathcal {C}'$. 
Due to the limited size of the local buffer $C$, the encoder uses a sequentially coding scheme as follows. 
First, the encoder receives $\mathcal{C}'_1$ and encode it using the shared key. Let the output be $\mathcal{C}_1$. The encoder sends $\mathcal{C}_1$ over the shared link and also stores it to the public cache. Then, the encoder receives the second block $\mathcal{C}'_2$ and encodes it using $W$ and $\mathcal{C}_1$. The output is denoted by $\mathcal{C}_2$ and is sent through the shared link and is stored to the public cache.  
Here, we assume that $L$ is a large quantity. Furthermore, the encoder encodes $i$-th block $\mathcal{C}'_i$ using $W$ and $(\mathcal{C}_1,\ldots,\mathcal{C}_{i-1})$. The encoding scheme continues until the last block of the codeword $\mathcal {C}'$. Let $\mathcal{C}=(\mathcal{C}_1,\mathcal{C}_2,\ldots)$ denote the delivered message consisting of the messages sent over the shared link.  
%The output of the encoder $\mathcal{C}(Y,W,d_1,\ldots,d_K)$ is the codeword the server sends over the shared link in order to satisfy the demands of the users $(Y_{d_1},\ldots,Y_{d_K})$. 
At the user side, user $k$ employs the decoding function $\mathcal{D}_k$ to recover its demand $Y_{d_k}$ without loss, i.e., $\hat{Y}_{d_k}=\mathcal{D}_k(Z_k,W,\mathcal{C},d_1,\ldots,d_K)$ and
\begin{align*}
\mathbb{P}(\mathcal{D}_k(\mathcal{C},W,Z_k,d_1,\ldots,d_K)\!=\!Y_{d_k})\!=\!1.
\end{align*}
This problem can be studied, and similar coding schemes can be applied. We leave this as a topic for future research.

 \section{conclusion}\label{concul}
 We have introduced a cache-aided compression problem with a perfect privacy constraint, where the information delivered over the shared link during the delivery phase is independent of underlying variable $X$ that is correlated with the files in the database that can be requested by the users. %or should leak only a limited amount of information about an underlying variable $X$ that is correlated with the files in the database that can be requested by the users. 
 We propose an achievable scheme using two-part code construction that benefits from the  %It has been shown that by using two-part construction coding the achievability scheme can be designed. To do so, we have used 
 Functional Representation Lemma to encode the second part of the code, which is constructed based on the specific user demands. For the first part of the code we hide the private data using one-time pad coding. We improved the obtained bounds by using greedy entropy-based algorithm and proposed a third achievability scheme considering two cases using common information concepts. The main reason the greedy algorithm improves the results is that it builds a random variable with minimum entropy while satisfying the constraints imposed by the Functional Representation Lemma (FRL). We have shown that the entropy-based algorithm can improve upon the code constructed using the Functional Representation Lemma (FRL), but at the cost of increased complexity. Finally, we have shown that the code based on common information can outperform the other two schemes in different scenarios.
\section*{Appendix A}
%\subsection*{Proof of Lemma~\ref{EFRL}:}
%The proof of \eqref{prove} is based on the construction used in \cite[Lemma~3]{king1} and \cite[Lemma~1]{kostala}. When $X$ is a deterministic function of $Y$, $\tilde{U}$ is found by FRL \cite[Lemma~1]{kostala} as it is used in 
%\cite[Lemma~3]{king1} satisfies $|\mathcal{\tilde{U}}|\leq |\mathcal{Y}|-|\mathcal{X}|+1$. Using the construction as in \cite[Lemma~3]{king1} we have $|\mathcal{U}|\leq (|\mathcal{Y}|-|\mathcal{X}|+1)(|\mathcal{X}|+1)$.
\subsection*{Proof of Theorem~\ref{th1}:}
	In the placement phase we use the same scheme as in \cite{maddah1}. In the delivery phase, we use the following strategy. Let $W$ be the shared secret key with key size $T=|\mathcal{X}|$, which is uniformly distributed over $\{1,,\ldots,T\}=\{1,\ldots,|\mathcal{X}|\}$ and independent of $(X,Y_1,\ldots,Y_N)$. As shown in Fig. \ref{achieve}, first, the private data $X$ is encoded using the shared secret key \cite[Lemma~1]{kostala2}. Thus, we have
\begin{align*}
\tilde{X}=X+W\ \text{mod}\ |\mathcal{X}|.
\end{align*}
Next, we show that $\tilde{X}$ has uniform distribution over $\{1,\ldots,|\mathcal{X}|\}$ and $I(X;\tilde{X})=0$. We have
\begin{align}\label{t}
H(\tilde{X}|X)\!=\!H(X\!+\!W|X)\!=\!H(W|X)\!=\!H(W)\!=\! \log(|\mathcal{X}|).
\end{align}
Furthermore, $H(\tilde{X}|X)\leq H(\tilde{X})$, and combining it with \eqref{t}, we obtain $$H(\tilde{X}|X)= H(\tilde{X})=\log(|\mathcal{X}|).$$ For encoding $\tilde{X}$ we use $\ceil{\log(|\mathcal{X}|)}$ bits. Let $\mathcal{C}'$ be the response that the server sends in the delivery phase for the scheme proposed in \cite[Theorem 1]{maddah1}, also given in \eqref{cache1}. We produce $U$ based on the construction proposed in \cite[Lemma~1]{kostala}, where we use $Y\leftarrow \mathcal{C}'$ and $X \leftarrow X$ in the FRL. Thus, 
\begin{align}
H(\mathcal{C}'|X,U)&=0,\label{kharkosde1}\\
I(U;X)&=0.\label{koonnane}
\end{align} 
Since the construction of $U$ is based on Lemma~\ref{aghabwi}, we have
\begin{align}
H(U)\leq \sum_x H(\mathcal{C}'|X=x)+1.
\end{align}
We encode $U$ and $\tilde{X}$ and denote them by $\mathcal{C}_2$ and $\mathcal{C}_1$, which have average lengths at most $\sum_x H(\mathcal{C}'|X=x)+1$ and $\ceil{\log(|\mathcal{X}|)}$, respectively. To encode $U$ we use any traditional lossless codes. Note that both $\mathcal{C}_2$ and $\mathcal{C}_1$ are lossless and variable length codes. Considering $\mathcal{C}=(\mathcal{C}_1,\mathcal{C}_2)$ we have
\begin{align}
\mathbb{L}(P_{XY},|\mathcal{X}|)\leq \sum_{x\in\mathcal{X}}H(\mathcal{C}'|X=x)+\!1+\!\ceil{\log (|\mathcal{X}|)}.
\end{align}
Next, we present the decoding strategy at the user side. Since $W$ is shared with each user, user $i$ decodes $X$ by using $\mathcal{C}_1$ and $W$, i.e., by using one-time-pad decoder. By adding $|\mathcal{X}|-W$ to $\tilde{X}$ we obtain $X$. Then, based on \eqref{kharkosde1} user $i$ decodes $\mathcal{C}'$ using $X$ and $U$. Finally, user $i$ decodes $Y_{d_i}$ using $\mathcal{C}'$ and its local cache contents $Z_i$, for more details see \cite[Theorem 1]{maddah1}.
What remains to be shown is the leakage constraint. We choose $W$ independent of $(X,U)$. We have
\begin{align*}
I(\mathcal{C};X)&=I(\mathcal{C}_1,\mathcal{C}_2;X)\\&=I(U,\tilde{X};X)\\&=I(U;X)+I(\tilde{X};X|U)\\&=I(\tilde{X};X|U)\\&\stackrel{(a)}{=}H(\tilde{X}|U)-H(\tilde{X}|X,U)\\&=H(\tilde{X}|U)-H(X+W|X,U)\\&\stackrel{(b)}{=}H(\tilde{X}|U)-H(W)\\&\stackrel{(c)}{=}H(\tilde{X})-H(W)\\&=\log(|\mathcal{X}|)\!-\!\log(|\mathcal{X}|)\\&=0
\end{align*}
where (a) follows from \eqref{koonnane}; (b) follows since $W$ is independent of $(X,U)$; and (c) from the independence of $U$ and $\tilde{X}$. The latter follows since we have
\begin{align*}
0&\leq I(\tilde{X};X|U) \\&= H(\tilde{X}|U)-H(W)\\&\stackrel{(i)}{=}H(\tilde{X}|U)-H(\tilde{X})\leq 0.
\end{align*}
Thus, $\tilde{X}$ and $U$ are independent. Step (i) above follows by the fact that $W$ and $\tilde{X}$ are uniformly distributed over $\{1,\ldots,|\mathcal{X}|\}$, i.e., $H(W)=H(\tilde{X})$. As a summary, if we choose $W$ independent of $(X,U)$ the leakage to the adversary is zero.
	To prove \eqref{koon2wi} and \eqref{gohwi} we use the same code, and note that when $|\mathcal{C}'|$ is finite we can use the bound in \eqref{prwi}, which results in \eqref{koon2wi}. Furthermore, when $X$ is a deterministic function of $\mathcal{C}'$ we can use the bound in \eqref{provewi} that leads to \eqref{gohwi}.
	\subsection*{Proof of Theorem~\ref{conv}:}
	To obtain \eqref{lower2}, let $U$ satisfy
	\begin{align}
	H(Y_{d_k}|U,W,Z_k)=0,\ \forall k\in[K].
	\end{align}
	For such a $U$ We have 
	\begin{align*}
	\mathbb{L}(P_{XY},T)&\geq \min_{U:I(U;X)=0} H(U) \geq \min H(U).
	\end{align*}
	Using the cutset bound in \cite[Theorem 2]{maddah1} we obtain
	\begin{align}
	\floor{\frac{N}{t}}H(U)+tMF+\log{T}\geq t\floor{\frac{N}{t}}F.
	\end{align}
	To obtain \eqref{lower2}, consider the first $t$ users with their local caches $Z_1,\ldots,Z_t$, let $C_i,\ i\in[\floor{\frac{N}{t}}]$, be the response to the demands $Y_{it+1},\ldots,Y_{it+t}$ of those users. Considering all responses $C_i,\ i\in[\floor{\frac{N}{t}}]$, and demands $Y_1,\ldots, Y_{t\floor{\frac{N}{t}}}\triangleq Y_1^{t\floor{\frac{N}{t}}}$ we have
	\begin{align*}
	&H(C_{1}^{\floor{\frac{N}{t}}})+sMF\geq H(C_{1}^{\floor{\frac{N}{t}}})+H(Z_{1}^{t}|C_{1}^{\floor{\frac{N}{t}}},X=x)\\&=H(C_{1}^{\floor{\frac{N}{t}}}|X=x)+H(Z_{1}^{t}|C_{1}^{\floor{\frac{N}{t}}},X=x)\\&=
	\!H(C_{1}^{\floor{\frac{N}{t}}}\!,Z_{1}^{t}|W,X=x)\\&\stackrel{(a)}{=}\! H(Y_1^{t\floor{\frac{N}{t}}},C_{1}^{\floor{\frac{N}{t}}}\!,Z_{1}^{t}|W,X=x)\\ &\stackrel{(b)}{\geq} H(Y_1^{t\floor{\frac{N}{t}}}|W,X=x)\\&=H(Y_1^{t\floor{\frac{N}{t}}}|X=x)
	\end{align*}
	where in (a) we used decodability condition $$H(Y_1^{t\floor{\frac{N}{t}}}|C_{1}^{\floor{\frac{N}{t}}},Z_{1}^{t},W)=0,$$ and (b) follows by independence of $W$ and $(X,Y)$. By taking maximum over $x$ and $t$ we obtain the result. 
	\subsection*{Proof of Theorem~\ref{th12}:}
	In the placement phase, we use the same scheme as discussed before. In the delivery phase, we use the following strategy.
	Similar to \cite{amircache}, we use two-part code construction to achieve the upper bounds. As shown in Fig. \ref{achieve2}, we first encode the private data $X$ using one-time pad coding \cite[Lemma~1]{kostala2}, which uses $\ceil{\log(|\mathcal{X}|)}$ bits. %The rest follows since in the one-time pad coding, the RV added to $X$ is the shared key, which is of size $|\mathcal{X}|$, and as a result the output has uniform distribution.
	Next, we produce $U$ based on greedy entropy-based algorithm proposed in \cite{kocaoglu2017entropic} which solves the minimum entropy problem in \eqref{minent2}, where $\mathcal{C}'$, defined in \eqref{cache1}, is the response that the server sends over the shared link to satisfy the users$'$ demands \cite{maddah1}. %Note that to produce such a $U$ we follow the construction used in Lemma \ref{aghabwi}. 
	Thus, we have
	\begin{align}
	H(\mathcal{C}'|X,U)&=0,\label{kharkosde12}\\
	I(U;X)&=0,
	\end{align}  
	Note that in Remark 6 we substitute $\mathcal{G}_S$ by $U$ and for binary $X$ we have
	\begin{align}\label{kiun2}
	H(U) \leq H(Q^*)+\frac{\log e}{e},
	\end{align}
	and for $|X|>2$,
	\begin{align}\label{kiun3}
	H(U) \leq H(Q^*)+\frac{1+\log e}{2}.
	\end{align}
	
	%Since we used the construction as in Lemma \ref{FRL}, 
	%and $U$ also satisfies \eqref{wi}. 
	Thus, we obtain \eqref{kiun} and \eqref{kiun1}.
	% 	\begin{align*}
	% 	\mathbb{L}(P_{XY},|\mathcal{X}|)\leq H(Q^*)+\frac{1+\log e}{2}+\!1+\!\ceil{\log (|\mathcal{X}|)}.
	% 	\end{align*}
	% 	To prove \eqref{koon2wi} we use the same coding scheme with the bound in \eqref{prwi}. If $X$ is a deterministic function of $\mathcal{C}'$, we can use the bound in \eqref{provewi} that leads to \eqref{gohwi}. 
	Moreover, for the leakage constraint we note that the randomness of one-time-pad coding is independent of $X$ and the output of the greedy entropy-based algorithm $U$.
	As shown in Fig. \ref{decode2}, at user side, each user, e.g., user $i$, first decodes $X$ using one-time-pad decoder. Then, based on \eqref{kharkosde12} it decodes $\mathcal{C}'$ using $U$ and $X$. Finally, it decodes $Y_{d_i}$ using local cache $Z_i$ and the response $\mathcal{C}'$.
	\subsection*{Proof of Theorem~\ref{loo}:}
	To achieve \eqref{log}, we use the solution to $h_0(P_{X\mathcal{C}'})=g_0(P_{X\mathcal{C}'})$ instead of the greedy entropy-based algorithm. 
	Furthermore, to achieve \eqref{ass} and \eqref{mass}, we use two-part construction coding and inequalities obtained in \cite[Theorem 1]{zero}. Upper bounds \eqref{koonwi11} and \eqref{koonwi22} are obtained in Theorem \ref{th1}. Finally, to achieve \eqref{kos3}, let the shared key $W$ be independent of $(X,Y)$ and has uniform distribution $\{1,,\ldots,T\}=\{1,\ldots,|\mathcal{C}'|\}$. We construct $\tilde{C}$ using one-time pad coding. We have
	\begin{align*}
	\tilde{C}=\mathcal{C}'+W\ \text{mod}\ |\mathcal{Y}|,
	\end{align*}
	where $\mathcal{C}'$ is defined in \eqref{cache1} and clearly we have
	\begin{align}
	I(\tilde{C};X)=0.
	\end{align}
	Then, $\tilde{C}$ is encoded using any lossless code which uses at most $\ceil{\log(|\mathcal{C}|)}$ bits. At decoder side, we first decode $\mathcal{C}'$ using the shared key. We then decode each demanded file by using the cache contents and $\mathcal{C}'$.
	 \clearpage
	\bibliographystyle{IEEEtran}
	{\balance \bibliography{IEEEabrv,IZS}}
\end{document}